\documentclass[pdflatex,sn-mathphys-num]{sn-jnl} 

\usepackage{graphicx}
\usepackage{tikz}
\usetikzlibrary{arrows.meta, positioning, shapes.geometric, decorations.markings}

\usetikzlibrary{calc, patterns, patterns.meta}
\usepackage{subcaption}

\usepackage{multirow}%
\usepackage{amsmath,amssymb,amsfonts}%
\usepackage{amsthm}%
\usepackage{mathrsfs}%
\usepackage[title]{appendix}%
\usepackage{xcolor}%
\usepackage{textcomp}%
\usepackage{manyfoot}%
\usepackage{booktabs}%
\usepackage{algorithm}%
\usepackage{algorithmicx}%
\usepackage{algpseudocode}%
\usepackage{listings}%

\theoremstyle{thmstyleone}%
\newtheorem{theorem}{Theorem}
\newtheorem{proposition}[theorem]{Proposition}%

\theoremstyle{thmstyletwo}%
\newtheorem{example}{Example}%
\newtheorem{remark}{Remark}%

\newtheorem{observation}{Observation}
\theoremstyle{thmstylethree}%
\newtheorem{definition}{Definition}%

\begin{document}

\title[Word-representability and comparability]{Word-representability and comparability: Minimal forbidden induced subgraphs and cover number bounds}

\author[1]{\fnm{Benny} \sur{George Kenkireth}}\email{ben@iitg.ac.in}
\author[1]{\fnm{Gopalan} \sur{Sajith}}\email{sajith@iitg.ac.in}
\author*[1]{\fnm{Sreyas} \sur{Sasidharan}}\email{sreyas.s@iitg.ac.in}

\affil*[1]{\orgdiv{Department of Computer Science and Engineering}, \orgname{Indian Institute of Technology Guwahati}, \orgaddress{\city{Guwahati}, \country{India}}}

\abstract{
Word-representable graphs, characterized by the existence of a semi-transitive orientation, form a well-studied class of graphs. Comparability graphs form another well-studied class and constitute a subclass of word-representable graphs. Both classes are hereditary and admit characterizations in terms of minimal forbidden induced subgraphs. While the minimal forbidden induced subgraphs for comparability graphs are completely characterized, the corresponding characterization for word-representable graphs remains open. 

In this paper, we precisely determine which minimal non-comparability graphs are also minimal non-word-representable graphs by classifying minimal non-comparability graphs according to whether they are word-representable. As a consequence, we provide a complete description of minimal non-word-representable graphs containing an all-adjacent vertex.

We also address an open problem posed by Kenkireth et al.\ concerning the cover number of word-representable graphs by comparability graphs. We demonstrate the existence of word-representable graphs on $n$ vertices whose cover number by comparability graphs is $\Omega(\log n)$, which establishes that the universal $O(\log n)$ upper bound is asymptotically tight for the class of word-representable graphs. For triangle-free circle graphs, we establish that the cover number by comparability graphs is at most $3$ and demonstrate that this bound is tight. More generally, we show that for any circle graph $G$ with clique number $\omega(G)$, the cover number by comparability graphs is bounded by $O(\log \omega(G))$. Finally, we identify four subclasses of word-representable graphs for which the cover number by comparability graphs of every graph in these classes is at most $2$.

}




\maketitle

\section{Introduction}\label{sec1}

Word-representable graphs have been extensively studied. 
The notion 
was introduced by Kitaev and Pyatkin in $2008$~\cite{onrepgraphs}. 
A graph is called \emph{word-representable} if there exists a word over its vertex set such that 
two vertices are adjacent in the graph if and only if 
they alternate in the word. 
The class of word-representable graphs encompasses many well-known graph families, including $3$-colorable graphs, comparability graphs, subcubic graphs, and circle graphs~\cite{Halldorsson2011,book}. 
Word-representable graphs admit a complete characterization in terms of semi-transitive orientations: a graph is word-representable if and only if it has such an orientation.
Informally, a \emph{semi-transitive orientation} is an acyclic orientation of the edges of the graph such that, 
whenever there is a directed path $v_1 \to v_2 \to \dots \to v_k$, either there is no edge from $v_1$ to $v_k$, 
or, if such an edge exists, then all edges $v_iv_j$ with $1\leq i< j \leq k$ are present and oriented as $v_i \to v_j$.  
This concept will be defined formally in Section~\ref{preliminaries}.

Word-representable graphs form a \emph{hereditary} class; i.e.,\ every induced subgraph of a word-representable graph is itself word-representable.  Every such class can be characterized by \emph{forbidden induced subgraphs} and has a unique set of minimal ones, here, in the context of word-representability, called \emph{minimal non-word-representable graphs}. Despite substantial progress (see, e.g., \cite{newres,kitaev2017comprehensive,onrepgraphs}), a complete description of the minimal non-word-representable graphs remains open.

Comparability graphs form an important subclass of word-representable graphs. 
They constitute a classical family of graphs
arising from partial orders: a graph is a comparability graph if its vertices can be endowed with a partial order such that two vertices are adjacent precisely when they are comparable. Equivalently, a graph is a comparability graph if its edges admit a transitive orientation.
Comparability graphs are also hereditary; i.e.,\ every induced subgraph of a comparability graph is itself a comparability graph. 
The minimal forbidden induced subgraphs of comparability graphs are referred to as \emph{minimal non-comparability graphs}, and the complete set of such graphs was determined by Gallai~\cite{Gallai1967TransitivOG}. 

Since comparability graphs form a subclass of word-representable graphs, it is natural to ask whether there exist graphs that are simultaneously minimal non-comparability and minimal non-word-representable, and to identify them. In this paper, we answer this question completely.
By classifying minimal non-comparability graphs into word-representable and non-word-representable cases, we determine exactly which of them are minimal non-word-representable. 
This classification also yields a complete list of minimal non-word-representable graphs containing a universal (all-adjacent) vertex, thereby identifying many infinite families of graphs that are minimal non-word-representable.
In the broader context of comparing word-representable graphs and comparability graphs, we also consider the problem of determining the minimum number of comparability graphs whose union forms a given word-representable graph, which was an open problem posed in \cite{kenkireth2026word}.
More formally, given a word-representable graph $G$, we consider the smallest integer $k$ for which there exist comparability graphs $G_1, G_2, \dots, G_k$ such that every edge of $G$ is contained in at least one of these graphs. That is, the edge set of $G$ is the union of the edge sets of $G_1, \dots, G_k$. This integer $k$ is referred to as the \emph{cover number} of $G$ by comparability graphs.

Studying both aspects—identifying common minimal forbidden induced subgraphs and analyzing the cover number—provides insight into the structural relationship between these two graph classes. 
Since comparability graphs form a subclass of word-representable graphs, both problems are natural and contribute to a deeper understanding of how these classes are related.

In this paper, we investigate the following problems.

\paragraph{Problem 1}
\emph{Determine which minimal non-comparability graphs are also minimal non-word-representable graphs.}

\medskip

\paragraph{Problem 2}
\emph{For an arbitrary word-representable graph, determine its cover number by comparability graphs.} 

This problem was posed in~\cite{kenkireth2026word}.

\paragraph{Our contributions}
We completely resolve Problem~1 and provide structural results for Problem~2. Concretely, our contributions are as follows:

\begin{enumerate}
 \item We provide a complete classification of minimal non-comparability graphs into word-representable and non-word-representable graphs, explicitly identifying the intersection between the classes of minimal non-comparability and minimal non-word-representable graphs.

 \item We characterize all minimal non-word-representable graphs that contain a universal (all-adjacent) vertex.

 \item We demonstrate the existence of word-representable graphs on $n$ vertices whose cover number by comparability graphs is $\Omega(\log n)$. Since any $n$-vertex graph can be covered by $O(\log n)$ comparability graphs, this upper bound is asymptotically tight for the class of word-representable graphs.

 \item For triangle-free circle graphs, we show that the cover number by comparability graphs is at most $3$ and that this bound is tight.

\item More generally, we establish that for any circle graph $G$ with clique number $\omega(G)$, the cover number by comparability graphs is bounded by $O(\log \omega(G))$.

\item We identify four subclasses of word-representable graphs for which the cover number by comparability graphs is at most $2$ for every graph in the class.

\end{enumerate}

The remainder of the paper is organized as follows. In Section~\ref{preliminaries}, we present the necessary preliminaries and background results used throughout the paper. Section~\ref{subsec-motivation} provides motivation for the problems studied. In Section~\ref{section-classification}, we address Problem~1 by classifying minimal non-comparability graphs according to whether they are word-representable. In Section~\ref{covering-section}, we investigate Problem~2 concerning coverings of word-representable graphs by comparability graphs. Section~\ref{conclusion-section} contains concluding remarks.

\section{Preliminaries}

\label{preliminaries}

All graphs considered in this paper are simple and undirected. For a graph $G$, let $V(G)$ and $E(G)$ denote its vertex set and edge set, respectively. For a subset $S \subseteq V(G)$, the subgraph of $G$ induced by $S$ is denoted by $G[S]$. The chromatic number, clique number, and girth of a graph $G$ are denoted by $\chi(G)$, $\omega(G)$, and $\mathrm{girth}(G)$, respectively, unless otherwise specified. 
All logarithms in this paper are taken to base $2$ unless otherwise specified.

\subsection{Word-Representable Graphs}
\begin{definition}
    Let $w$ be a word over an alphabet $\Sigma$, and let $a,b \in \Sigma$.  
    Define $w|_{ab}$ as the subsequence of $w$ obtained by retaining only the occurrences of $a$ and $b$, in their original order. The letters $a$ and $b$ \emph{alternate} in $w$ if $w|_{ab}$ contains neither the substring $aa$ nor $bb$.

\end{definition}

\begin{example}
    For $w = 312314$, the letters $1$ and $2$ alternate, since $w|_{12}=121$ contains no substring of the form $11$ or $22$.  
    In contrast, $1$ and $4$ do not alternate, as $w|_{14}=114$ contains a repetition.  
    (Other pairs can be checked similarly; these two illustrate the concept.)
\end{example}

\begin{definition}
\label{wrgdefn}
    A graph $G$ is \emph{word-representable} if there exists a word $w$ over $V(G)$ such that for all distinct $x,y \in V(G)$:
    \[
        (x,y) \in E(G) \iff x \text{ and } y \text{ alternate in } w.
    \]
    Such a word $w$ is called a \emph{word-representant} of $G$.
\end{definition}

\begin{figure}[htbp]
    \centering
    \subfloat[$G_{1}$]{
        \begin{tikzpicture}[scale=0.4, transform shape, every node/.style={font=\Huge}]
            \tikzset{vertex/.style = {circle, fill=black, text=white, minimum size=0.4cm}}
            \tikzset{edge/.style = {-, thick}} 

           
            \node[vertex, label=above:{$1$}] (1) at (90:2cm) {};
            \node[vertex, label=left:{$2$}] (2) at (162:2cm) {};
            \node[vertex, label=left:{$3$}] (3) at (234:2cm) {};
            \node[vertex, label=right:{$4$}] (4) at (306:2cm) {};
            \node[vertex, label=right:{$5$}] (5) at (378:2cm) {};

           
            \draw[edge] (1) -- (2);
            \draw[edge] (2) -- (3);
            \draw[edge] (3) -- (4);
            \draw[edge] (5) -- (4);
            \draw[edge] (5) -- (1);
          
        \end{tikzpicture}
        \label{fig:g1}
    }
    \hspace{2cm}
    \subfloat[$G_{2}$]{
     \begin{tikzpicture}[scale=0.4, transform shape, every node/.style={font=\Huge}]
    \tikzset{vertex/.style = {circle, fill=black, text=white, minimum size=0.4cm}}
            \tikzset{edge/.style = {-, thick}} 

    \node[vertex,label=below:{$1$}] (1) at (0,0) {};
    \node[vertex,label=above:{$2$}] (2) at (90:2cm) {};
    \node[vertex,label=left:{$3$}] (3) at (162:2cm) {};
    \node[vertex,label=left:{$4$}] (4) at (234:2cm) {};
    \node[vertex,label=right:{$5$}] (5) at (306:2cm) {};
    \node[vertex,label=right:{$6$}] (6) at (378:2cm) {};

    \draw[edge] (1) -- (2);
    \draw[edge] (1) -- (3);
    \draw[edge] (1) -- (4);
    \draw[edge] (6) -- (1);
    \draw[edge] (5) -- (1);


    \draw[edge] (2) -- (3);
    \draw[edge] (2) -- (6);
    \draw[edge] (6) -- (5);
    \draw[edge] (5) -- (4);
    \draw[edge] (4) -- (3);
\end{tikzpicture}
        \label{fig:g2}
    }
    \caption{Graph $G_{1}$ is word-representable , while $G_{2}$ is non-word-representable.}
    \label{fig:combined_word_rep_non_word_rep}
\end{figure}
\begin{example}
    The graph $G_{1}$ in Figure~\ref{fig:combined_word_rep_non_word_rep} is word-representable, with the word $w=1521324354$ serving as a word-representant of $G_{1}$. 
The graph $G_{2}$ in Figure~\ref{fig:combined_word_rep_non_word_rep} is the wheel graph $W_{5}$ on six vertices, which is a well-known non-word-representable graph.
\end{example}

Word-representable graphs admit a characterization via the notion of a \emph{semi-transitive orientation}, defined below.
\begin{definition}
\label{semi-trans-defn}
A graph $G$ is \emph{semi-transitive} if it admits an acyclic orientation such that for every directed path
$v_1 \rightarrow v_2 \rightarrow \cdots \rightarrow v_k$, either
$v_1 \not\rightarrow v_k$, or
$v_i \rightarrow v_j$ exists for all $1 \le i < j \le k$.
If $v_1 \rightarrow v_k$ is present while some
$v_i \rightarrow v_j$ with $i < j$ is missing, we call this a
\emph{shortcut}~\cite{book}.
\end{definition}

\begin{figure}[htbp]
    \centering
    \begin{tikzpicture}[scale=0.4, transform shape, every node/.style={font=\Huge}]
        \tikzset{vertex/.style={circle, fill=black, text=white, minimum size=0.4cm}}
        \tikzset{edge/.style={
            thick,
            postaction={decorate},
            decoration={
                markings,
                mark=at position 0.5 with {\arrow{>}}
            }
        }}

     
        \node[vertex, label=above:{$1$}] (1) at (90:2cm) {};
        \node[vertex, label=left:{$2$}] (2) at (162:2cm) {};
        \node[vertex, label=left:{$3$}] (3) at (234:2cm) {};
        \node[vertex, label=right:{$4$}] (4) at (306:2cm) {};
        \node[vertex, label=right:{$5$}] (5) at (378:2cm) {};

        \draw[edge] (1) -- (2);
        \draw[edge] (2) -- (3);
        \draw[edge] (3) -- (4);
        \draw[edge] (5) -- (4);
        \draw[edge] (1) -- (5);
        
    \end{tikzpicture}
    \caption{A semi-transitive orientation of the word-representable graph $G_{1}$ shown in Figure \ref{fig:combined_word_rep_non_word_rep}}
    \label{fig:g1-directed}
\end{figure}

\begin{example}
Figure~\ref{fig:g1-directed} illustrates an orientation of the graph $G_1$ (shown in Figure~\ref{fig:combined_word_rep_non_word_rep}) that is semi-transitive, thereby providing an alternative verification of its word-representability. We can verify that this orientation satisfies the conditions of Definition~\ref{semi-trans-defn}. The longest directed path in this graph is $1 \rightarrow 2 \rightarrow 3 \rightarrow 4$, which does not have its endpoints connected, hence no shortcut. Further, the directed path $1 \rightarrow 5 \rightarrow 4$ is a path of length $2$, and the shortcut condition is vacuous for paths of length at most $2$. Moreover, the orientation is acyclic. Hence, the graph is semi-transitive.

\end{example}

The following theorem provides the fundamental characterization of word-representable graphs.

\begin{theorem}[\cite{book}]
\label{wrg=semi}
A graph \(G\) is word-representable if and only if it admits a semi-transitive orientation.
\end{theorem}

The next result refines this characterization by showing that a prescribed vertex can be chosen as a source in a semi-transitive orientation.

\begin{theorem}[\cite{kitaev2023humanverifiable}]
\label{semi-v-source}
Suppose that a graph $G$ is word-representable, and $v$ is a vertex in $G$. Then, there exists a semi-transitive orientation of $G$, where $v$ is a source.
\end{theorem}

The following theorem gives a sufficient condition for semi-transitivity in terms of graph colorability.

\begin{theorem}[\cite{Halldorsson2011}]
\label{3-col-semi}
Any $3$-colorable graph is semi-transitive.
\end{theorem}
Word-representable graphs constitute a hereditary graph class. We recall the formal definition of a hereditary class below.

\begin{definition}
A class $X$ of graphs is called \emph{hereditary} if, for every graph $G \in X$, every induced subgraph of $G$ also belongs to $X$.
\end{definition}

Examples of hereditary graph classes include planar graphs, comparability graphs, and word-representable graphs~\cite{book}. Any hereditary class can be characterized by its set of minimal forbidden induced subgraphs, which capture the essential structural constraints defining the class.

\begin{definition}
A graph $G$ is called a \emph{minimal forbidden induced subgraph} for a hereditary class $X$ if $G \notin X$, but every proper induced subgraph of $G$ belongs to $X$.
\end{definition}

\begin{definition}
A graph is called \emph{minimal non-word-representable} if it is a minimal forbidden induced subgraph for the class of word-representable graphs.
\end{definition}

\begin{example}
    The wheel graph $G_2$ shown in Figure~\ref{fig:combined_word_rep_non_word_rep} is a minimal non-word-representable graph.
\end{example} Determining the complete set of minimal non-word-representable graphs remains an open problem.

\subsection{Comparability Graphs}

\begin{definition}
A graph $G$ is called a \emph{comparability graph} if there exists an orientation of its edges such that the resulting directed graph is transitive; that is, whenever $a \to b$ and $b \to c$ are directed edges, it follows that $a \to c$ is also a directed edge. If no such orientation exists, then $G$ is said to be a \emph{non-comparability graph}.

Equivalently, a graph is a
comparability graph if its vertices can be endowed with a partial order such that two
vertices are adjacent precisely when they are comparable. 
\end{definition}

\begin{figure}[htbp]
\centering

\subfloat[]{
\begin{tikzpicture}[scale=0.44, transform shape, every node/.style={font=\Large}]
    \tikzset{vertex/.style = {shape=circle}}
    
    \node[vertex,fill=black,text=white,minimum size=0.35cm,label=left:{$1$}] (1) at (-1.5,-1.5) {};
    \node[vertex,fill=black,text=white,minimum size=0.35cm,label=above:{$2$}] (2) at (0.15,0) {};
    \node[vertex,fill=black,text=white,minimum size=0.35cm,label=right:{$3$}] (3) at (1.5,-1.5) {};
    \node[vertex,fill=black,text=white,minimum size=0.35cm,label=left:{$4$}] (4) at (-1.5,-3.5) {};
    \node[vertex,fill=black,text=white,minimum size=0.35cm,label=right:{$5$}] (5) at (1.5,-3.5) {};

    \draw[thick] (1) -- (2); 
    \draw[thick] (2) -- (3);
    \draw[thick] (3) -- (5);
    \draw[thick] (1) -- (4);      
    \draw[thick] (5) -- (4);    
    \draw[thick] (1) -- (3);  
\end{tikzpicture}
}
\hfill
\subfloat[]{
\begin{tikzpicture}[scale=0.44, transform shape, every node/.style={font=\Large}]
    \tikzset{vertex/.style = {shape=circle}}

    \node[vertex,fill=black,text=white,minimum size=0.35cm,label=left:{$1$}] (1) at (-1.5,-1.5) {};
    \node[vertex,fill=black,text=white,minimum size=0.35cm,label=above:{$2$}] (2) at (0.15,0) {};
    \node[vertex,fill=black,text=white,minimum size=0.35cm,label=right:{$3$}] (3) at (1.5,-1.5) {};
    \node[vertex,fill=black,text=white,minimum size=0.35cm,label=left:{$4$}] (4) at (-1.5,-3.5) {};
    \node[vertex,fill=black,text=white,minimum size=0.35cm,label=right:{$5$}] (5) at (1.5,-3.5) {};

    \draw[thick] (1) -- (2); 
    \draw[thick] (2) -- (3);
    \draw[thick] (3) -- (5);
    \draw[thick] (1) -- (4);      
    \draw[thick] (5) -- (4);    
\end{tikzpicture}
}
\hfill
\subfloat[]{
\begin{tikzpicture}[scale=0.44, transform shape, every node/.style={font=\Large}]
    \tikzset{
        vertex/.style = {shape=circle},
        edge/.style = {
            thick,
            postaction={decorate},
            decoration={
                markings,
                mark=at position 0.5 with {\arrow{>}}
            }
        }
    }

    \node[vertex,fill=black,text=white,minimum size=0.35cm,label=left:{$1$}] (1) at (-1.5,-1.5) {};
    \node[vertex,fill=black,text=white,minimum size=0.35cm,label=above:{$2$}] (2) at (0.15,0) {};
    \node[vertex,fill=black,text=white,minimum size=0.35cm,label=right:{$3$}] (3) at (1.5,-1.5) {};
    \node[vertex,fill=black,text=white,minimum size=0.35cm,label=left:{$4$}] (4) at (-1.5,-3.5) {};
    \node[vertex,fill=black,text=white,minimum size=0.35cm,label=right:{$5$}] (5) at (1.5,-3.5) {};

    \draw[edge] (2) -- (1);
    \draw[edge] (2) -- (3);
    \draw[edge] (3) -- (1);
    \draw[edge] (3) -- (5);
    \draw[edge] (4) -- (5);
    \draw[edge] (4) -- (1);
\end{tikzpicture}
}

\caption{(a) The graph $G_1$, which is a comparability graph. 
(b) The graph $G_2$, which is not a comparability graph. 
(c) A transitive orientation of $G_1$.}
\label{fig:compandnon-comp}
\end{figure}
\begin{example}
Graph $G_1$ in Figure~\ref{fig:compandnon-comp} is a comparability graph, as it admits the transitive orientation illustrated in Figure~\ref{fig:compandnon-comp}(c), whereas $G_2$ is not a comparability graph.
\end{example}

Note that any bipartite graph is a comparability graph. The following theorem establishes a connection between word-representable graphs and comparability graphs.

\begin{theorem}[\cite{onrepgraphs}]
\label{lemma_comp_wrg}
Let $G$ be a graph on $n$ vertices, and let $x \in V(G)$ be a vertex of degree $n-1$ (that is, $x$ is adjacent to every other vertex of $G$). Let $H = G - x$ be the induced subgraph obtained by deleting $x$. Then $G$ is word-representable if and only if $H$ is a comparability graph.
\end{theorem}

Comparability graphs form a hereditary class. 

\begin{definition}
A \emph{minimal forbidden induced subgraph} of a comparability graph is called a \emph{minimal non-comparability graph}; that is, it is a graph which is not a comparability graph, but all of its proper induced subgraphs are comparability graphs.
\end{definition}

\begin{example}An example of a minimal non-comparability graph is the cycle $C_5$, which is depicted as $G_2$ in Figure~\ref{fig:compandnon-comp}. Since $C_5$ does not admit a transitive orientation, it is not a comparability graph, while all of its proper induced subgraphs are comparability graphs.
\end{example}

The class of all minimal non-comparability graphs was characterized by Gallai~\cite{Gallai1967TransitivOG}. In his work, the classification is presented in two parts: one part consists of explicitly described graphs and graph families, while the other describes certain graphs as complements of specified graphs or families. In this paper, we explicitly compute these complements and present a unified list that includes both individual graphs and graph families. For some families, schematic representations are used for clarity, since explicit drawings can be difficult to convey. The unified list of all minimal non-comparability graphs is illustrated in Figure~\ref{fig:min-non-comp-all}.

\begin{figure}[htbp]
    \centering
    \subfloat[$G_{n}^{1}$ $;$ $n \geq 2$]{
\begin{tikzpicture}[scale=0.27, transform shape, every node/.style={font=\Huge}]

\node[draw, fill=black, circle,minimum size=0.4cm,  text=white, label=right:{$1$}] (1) at (2, 0) {};
\node[draw, fill=black, circle, minimum size=0.4cm, text=white, label=right:{$2$}] (2) at (3.5, 1.65) {};
\node[draw, fill=black, circle, minimum size=0.4cm, text=white, label=right:{$3$}] (3) at (2.9, 3.7) {};
\node[draw, fill=black, circle, minimum size=0.4cm, text=white] (r) at (0.5, 5.5) {};
\node[draw, fill=black, circle, minimum size=0.4cm, text=white] (r1) at (-1.9, 3.7) {};
\node[draw, fill=black, circle, minimum size=0.4cm, text=white, label=left:{$2n$}] (2n) at (-2.5, 1.65) {};
\node[draw, fill=black, circle, minimum size=0.4cm, text=white, label=left:{$2n+1$}] (2n+1) at (-1, 0) {};

\draw[-] (1) -- (2);
\draw[-] (2) -- (3);
\draw[-] (3) -- (r) node[midway, sloped, above, yshift=5pt] { . . . . . .}; 
\draw[-] (r) -- (r1);
\draw[-] (r1) -- (2n);
\draw[-] (2n) -- (2n+1);
\draw[-] (2n+1) -- (1);

\end{tikzpicture}

 \label{g1n}
}
\hspace{0.15cm}
\subfloat[$G_{n}^{2}$ $;$ $n \geq 2$]{
 
\begin{tikzpicture}[scale=0.27, transform shape, every node/.style={font=\Huge}]

\node[draw, fill=black, circle,minimum size=0.4cm,  text=white, label=below:{$1$}] (1) at (0, -2) {};
\node[draw, fill=black, circle,minimum size=0.4cm,  text=white, label=right:{$2$}] (2) at (2.8, 0) {};
\node[draw, fill=black, circle,minimum size=0.4cm,  text=white, label=left:{$2n+1$}] (2n+1) at (-2.8, 0) {};
\node[draw, fill=black, circle,minimum size=0.4cm,  text=white, label=right:{$3$}] (3) at (3, 2.3) {};
\node[draw, fill=black, circle,minimum size=0.4cm,  text=white, label=right:{$4$}] (4) at (1.5, 4) {};
\node[draw, fill=black, circle,minimum size=0.4cm,  text=white] (r) at (-1.5, 4) {};
\node[draw, fill=black, circle,minimum size=0.4cm,  text=white, label=left:{$2n$}] (2n) at (-3, 2.3) {};
\node[draw, fill=black, circle,minimum size=0.4cm,  text=white, label=left:{$x$}] (x) at (-4, -2) {};
\node[draw, fill=black, circle,minimum size=0.4cm,  text=white, label=right:{$y$}] (y) at (4, -2) {};

\draw[-] (1) -- (2);
\draw[-] (2) -- (3);
\draw[-] (3) -- (4);
\draw[-] (4) -- (r) node[midway, above, yshift=5pt] {. . . . . . .};
\draw[-] (r) -- (2n);
\draw[-] (2n) -- (2n+1);
\draw[-] (2n+1) -- (1);
\draw[-] (1) -- (3);
\draw[-] (1) -- (4);
\draw[-] (1) -- (4);
\draw[-] (1) -- (r);
\draw[-] (1) -- (2n);
\draw[-] (1) -- (2n+1);
\draw[-] (2n+1) -- (x);
\draw[-] (2) -- (y);

\end{tikzpicture}
    \label{g2n}
}
\hspace{0.15cm}
\subfloat[$G_{n}^{3}$ $;$ $n \geq 3$]{
     
\begin{tikzpicture}[scale=0.26, transform shape, every node/.style={font=\Huge}]

\node[draw, fill=black, circle,minimum size=0.4cm,  text=white, label=right:{$1$}] (1) at (1.15, -2.5) {};
\node[draw, fill=black, circle,minimum size=0.4cm,  text=white, label=right:{$2$}] (2) at (3.5, -0.25) {};
\node[draw, fill=black, circle,minimum size=0.4cm,  text=white, label=left:{$2n-1$}] (2n-1) at (-4, -0.25) {};
\node[draw, fill=black, circle,minimum size=0.4cm,  text=white, label=right:{$3$}] (3) at (3.37, 2.3) {};
\node[draw, fill=black, circle,minimum size=0.4cm,  text=white] (r1) at (-3.9, 2.3) {};
\node[draw, fill=black, circle,minimum size=0.4cm,  text=white, label=right:{$4$}] (4) at (1.2, 4) {};
\node[draw, fill=black, circle,minimum size=0.4cm,  text=white] (r) at (-1.7, 4) {};
\node[draw, fill=black, circle,minimum size=0.4cm,  text=white, label=left:{$2n$}] (2n) at (-1.65, -2.5) {};
\node[draw, fill=black, circle,minimum size=0.4cm,  text=white, label=left:{$x$}] (x) at (-4.6, -4.9) {};
\node[draw, fill=black, circle,minimum size=0.4cm,  text=white, label=right:{$y$}] (y) at (4.1, -4.9) {};

\draw[-] (1) -- (2);
\draw[-] (2) -- (3);
\draw[-] (3) -- (4);
\draw[-] (4) -- (r) node[midway, above, yshift=5pt] {. . . . . . .};
\draw[-] (1) -- (3);
\draw[-] (1) -- (4);
\draw[-] (1) -- (4);
\draw[-] (1) -- (r);
\draw[-] (1) -- (r);
\draw[-] (1) -- (r1);
\draw[-] (1) -- (2n-1);
\draw[-] (1) -- (r);
\draw[-] (2n) -- (2);
\draw[-] (2n) -- (3);
\draw[-] (2n) -- (4);
\draw[-] (2n) -- (r);
\draw[-] (2n) -- (r1);
\draw[-] (2n) -- (2n-1);
\draw[-] (2n) -- (2);
\draw[-] (r) -- (r1);
\draw[-] (2n-1) -- (r1);
\draw[-] (2n-1) -- (x);
\draw[-] (2) -- (y);
\draw[-] (2n) -- (y);
\draw[-] (1) -- (x);

\end{tikzpicture}

 \label{g3n}
}
 \hspace{0.15cm}
\subfloat[$G_{n}^{4}$ $;$ $n \geq 3$]{
  
  \begin{tikzpicture}[scale=0.26, transform shape, every node/.style={font=\Huge}]

\node[draw, fill=black, circle,minimum size=0.4cm,  text=white, label=right:{$1$}] (1) at (1.15, -2.5) {};
\node[draw, fill=black, circle,minimum size=0.4cm,  text=white, label=right:{$2$}] (2) at (3.5, -0.15) {};
\node[draw, fill=black, circle,minimum size=0.4cm,  text=white, label=left:{$2n$}] (2n) at (-4, -0.15) {};
\node[draw, fill=black, circle,minimum size=0.4cm,  text=white, label=right:{$3$}] (3) at (3.37, 2.3) {};
\node[draw, fill=black, circle,minimum size=0.4cm,  text=white] (r1) at (-3.9, 2.3) {};
\node[draw, fill=black, circle,minimum size=0.4cm,  text=white, label=right:{$4$}] (4) at (1.2, 4) {};
\node[draw, fill=black, circle,minimum size=0.4cm,  text=white] (r) at (-1.7, 4) {};
\node[draw, fill=black, circle,minimum size=0.4cm,  text=white, label=left:{$2n+1$}] (2n+1) at (-1.65, -2.5) {};
\node[draw, fill=black, circle,minimum size=0.4cm,  text=white, label=left:{$x$}] (x) at (-4.55, -4.9) {};
\node[draw, fill=black, circle,minimum size=0.4cm,  text=white, label=right:{$y$}] (y) at (4.05, -4.9) {};

\draw[-] (1) -- (2);
\draw[-] (2) -- (3);
\draw[-] (3) -- (4);
\draw[-] (4) -- (r) node[midway, above, yshift=5pt] {. . . . . . .};
\draw[-] (1) -- (3);
\draw[-] (1) -- (4);
\draw[-] (1) -- (4);
\draw[-] (1) -- (r);
\draw[-] (1) -- (r);
\draw[-] (1) -- (r1);
\draw[-] (1) -- (2n);
\draw[-] (1) -- (r);
\draw[-] (1) -- (2n+1);
\draw[-] (2n+1) -- (2);
\draw[-] (2n+1) -- (3);
\draw[-] (2n+1) -- (4);
\draw[-] (2n+1) -- (r);
\draw[-] (2n+1) -- (r1);
\draw[-] (2n+1) -- (2n);
\draw[-] (2n+1) -- (2);
\draw[-] (r) -- (r1);
\draw[-] (2n) -- (r1);
\draw[-] (2n) -- (x);
\draw[-] (2) -- (y);
\draw[-] (2n+1) -- (y);
\draw[-] (1) -- (x);

\end{tikzpicture}
 \label{g4n}
}
    \vspace{0.1cm}
    
    \subfloat[$G_{n}^{5}$ $;$ $n \geq 3$]{
        \centering
        \begin{tikzpicture}[scale=0.3, transform shape, every node/.style={font=\Huge}]
            \node[draw, fill=black, circle,minimum size=0.4cm,  text=white, label=right:{$a_{1}$}] (a1) at (0, 4) {};
            \node[draw, fill=black, circle,minimum size=0.4cm,  text=white, label=right:{$b_{1}$}] (b1) at (2.2, 2.6) {};
            \node[draw, fill=black, circle,minimum size=0.4cm,  text=white, label=right:{$a_{2}$}] (a2) at (2.95, 0.75) {};
            \node[draw, fill=black, circle,minimum size=0.4cm,  text=white, label=right:{$b_{2}$}] (b2) at (1.7, -0.7) {};
            \node[draw, fill=black, circle,minimum size=0.4cm,  text=white] (r) at (0, -2) {};
            \node[below= 0.1cm of r] {....};
            \node[draw, fill=black, circle,minimum size=0.4cm,  text=white] (r1) at (-1.7, -0.7) {};
            \node[draw, fill=black, circle,minimum size=0.4cm,  text=white, label=left:{$a_{n}$}] (an) at (-2.95, 0.75) {};
            \node[draw, fill=black, circle,minimum size=0.4cm,  text=white, label=left:{$b_{n}$}] (bn) at (-2.2, 2.6) {};
            
            \draw[-, dashed] (a1) -- (b1);
            \draw[-, dashed] (a2) -- (b1);
            \draw[-, dashed] (a2) -- (b2);
            \draw[-, dashed] (r) -- (b2);
            \draw[-, dashed] (r1) -- (r);
            \draw[-, dashed] (r1) -- (an);
            \draw[-, dashed] (an) -- (bn);
            \draw[-, dashed] (a1) -- (bn);
        \end{tikzpicture}
        \label{g5n}
    }
    \hspace{2cm}
    \subfloat[$G_{n}^{6}$ $;$ $n \geq 3$]{
        \centering
        \begin{tikzpicture}[scale=0.3, transform shape, every node/.style={font=\Huge}]
            \node[draw, fill=black, circle,minimum size=0.4cm,  text=white, label=right:{$a_{1}$}] (a1) at (0, 4) {};
            \node[draw, fill=black, circle,minimum size=0.4cm,  text=white, label=right:{$b_{1}$}] (b1) at (2.2, 2.6) {};
            \node[draw, fill=black, circle,minimum size=0.4cm,  text=white, label=right:{$a_{2}$}] (a2) at (2.95, 0.75) {};
            \node[draw, fill=black, circle,minimum size=0.4cm,  text=white, label=right:{$b_{2}$}] (b2) at (1.7, -0.7) {};
            \node[draw, fill=black, circle,minimum size=0.4cm,  text=white] (r) at (0, -2) {};
            \node[below= 0.1cm of r] {....};
            \node[draw, fill=black, circle,minimum size=0.4cm,  text=white, label=left:{$a_{n}$}] (an) at (-1.7, -0.7) {};
            \node[draw, fill=black, circle,minimum size=0.4cm,  text=white, label=left:{$b_{n}$}] (bn) at (-2.95, 0.75) {};
            \node[draw, fill=black, circle,minimum size=0.4cm,  text=white, label=left:{$c_{0}$}] (c0) at (-2.2, 2.6) {};
            
            \draw[-, dashed] (a1) -- (c0);
            \draw[-, dashed] (bn) -- (c0);
            \draw[-, dashed] (a2) -- (b1);
            \draw[-, dashed] (a2) -- (b2);
            \draw[-, dashed] (r) -- (b2);
            \draw[-, dashed] (an) -- (r);
            \draw[-, dashed] (bn) -- (an);
            \draw[-, dashed] (a1) -- (b1);
        \end{tikzpicture}
        \label{g6n}
    }

\vspace{0.1cm}
   
    \subfloat[$G_{n}^{7}$ $;$ $n \geq 1$]{
        \centering
       \begin{tikzpicture}[scale=0.34, transform shape, every node/.style={font=\Huge}]
            \node[draw, fill=black, circle, minimum size=0.4cm,  text=white, label=above:{$a$}] (a) {};
            \node[draw, fill=black, circle, minimum size=0.4cm, right=of a, text=white, label=above:{$b$} ] (b) {};
            \node[draw, fill=black, circle, minimum size=0.4cm, right=of b, text=white, label=above:{$c$} ] (c) {};
            \node[draw, fill=black, circle, minimum size=0.4cm, right=of c, text=white, label=above:{$d$} ] (d) {};

            \node[draw,fill=black, circle, minimum size=0.4cm, below left=of a, xshift= 1.1cm, yshift=-2cm, label=below:{$1$} ] (1) {};
            \node[draw, circle,  below left=of b, minimum size=0.4cm, xshift= 1.1cm,  yshift=-2cm, fill=black,  text=white, label=below:{$2$}] (2) {};

            \node[draw, circle, below left=of c, minimum size=0.4cm, xshift= 1.1cm, yshift=-2cm, fill=black,  text=white] (r) {}; 

            \node[above= 0.1cm of r] {....}; 

            \node[draw, circle, below left=of d, minimum size=0.4cm, xshift= 1.1cm, yshift=-2cm, fill=black,  text=white, label=below:{$n \mathnormal{-} 1$}] (n-1) {}; 

            \node[draw, circle, below=of d, minimum size=0.4cm, xshift= 1.1cm, yshift=-1.9cm, fill=black,  text=white, label=below:{$n$}] (n) {}; 

            \node[draw, circle,  below=of r, minimum size=0.4cm, yshift=-2cm, fill=black,  text=white, label=below:{$x$}] (x) {};

            \draw[-, dashed] (1) -- (2);
            \draw[-, dashed] (n-1) -- (r);
            \draw[-, dashed] (2) -- (r);
            \draw[-, dashed] (n-1) -- (n);
            \draw[-, dashed] (a) -- (b);
            \draw[-, dashed] (b) -- (c);
            \draw[-, dashed] (c) -- (d);

            \draw[-] (x) -- (1);
            \draw[-] (x) -- (2);
            \draw[-] (x) -- (n-1);
            \draw[-] (x) -- (n);
            \draw[-] (x) -- (r);

            \draw[-] (a) -- (n);
            \draw[-] (a) -- (2);
            \draw[-] (a) -- (n-1);
            \draw[-] (a) -- (r);

            \draw[-] (d) -- (1);
            \draw[-] (d) -- (2);
            \draw[-] (d) -- (n-1);
            \draw[-] (d) -- (r);



  \draw[-, bend left=90] (x) to (a);  
           \draw[-, bend right=90] (x) to (d);

        \end{tikzpicture}
        \label{g7n}
    }
    \hspace{0.24cm}
     \subfloat[$G_{n}^{8}$ $;$ $n \geq 1$]{
        \centering
        \begin{tikzpicture}[scale=0.34, transform shape, every node/.style={font=\Huge}]
            \node[draw, fill=black, circle, minimum size=0.4cm,  text=white, label=above:{$a$}] (a) {};
            \node[draw, fill=black, circle, minimum size=0.4cm, right=of a, text=white, label=above:{$b$} ] (b) {};
            \node[draw, fill=black, circle, minimum size=0.4cm, right=of b, text=white, label=above:{$c$} ] (c) {};
            \node[draw, fill=black, circle, minimum size=0.4cm, right=of c, text=white, label=above:{$d$} ] (d) {};

            \node[draw,fill=black, circle, minimum size=0.4cm, below left=of a, xshift= 1.1cm, yshift=-2cm, label=below:{$1$} ] (1) {};
            \node[draw, circle,  below left=of b, minimum size=0.4cm, xshift= 1.1cm,  yshift=-2cm, fill=black,  text=white, label=below:{$2$}] (2) {};

            \node[draw, circle, below left=of c, minimum size=0.4cm, xshift= 1.1cm, yshift=-2cm, fill=black,  text=white] (r) {}; 

            \node[above= 0.1cm of r] {....}; 

            \node[draw, circle, below left=of d, minimum size=0.4cm, xshift= 1.1cm, yshift=-2cm, fill=black,  text=white, label=below:{$n \mathnormal{-} 1$}] (n-1) {}; 

            \node[draw, circle, below=of d, minimum size=0.4cm, xshift= 1.1cm, yshift=-1.9cm, fill=black,  text=white, label=below:{$n$}] (n) {}; 

            \node[draw, circle,  below=of r, minimum size=0.4cm, yshift=-2cm, fill=black,  text=white, label=below:{$x$}] (x) {};

            \draw[-, dashed] (1) -- (2);
            \draw[-, dashed] (n-1) -- (r);
            \draw[-, dashed] (2) -- (r);
            \draw[-, dashed] (n-1) -- (n);
            \draw[-, dashed] (a) -- (b);
            \draw[-, dashed] (c) -- (d);

            \draw[-] (x) -- (1);
            \draw[-] (x) -- (2);
            \draw[-] (x) -- (n-1);
            \draw[-] (x) -- (n);
            \draw[-] (x) -- (r);

            \draw[-] (a) -- (n);
            \draw[-] (a) -- (2);
            \draw[-] (a) -- (n-1);
            \draw[-] (a) -- (r);

            \draw[-] (d) -- (1);
            \draw[-] (d) -- (2);
            \draw[-] (d) -- (n-1);
            \draw[-] (d) -- (r);

            \draw[-, bend left=90] (x) to (a);  
            \draw[-, bend right=90] (x) to (d); 

        \end{tikzpicture}
        \label{g8n}
    }
\hspace{0.24cm}
    \subfloat[$G_{n}^{9}$ $;$ $n \geq 2$]{
        \centering
        \begin{tikzpicture}[scale=0.34, transform shape, every node/.style={font=\Huge}]
            \node[draw, circle,  minimum size=0.4cm, fill=black,  text=white, label=above:{$a$}] (a) {};
            \node[draw, circle,  right=of a, minimum size=0.4cm, fill=black,  text=white, label=above:{$1$}] (1) {};
            \node[draw, circle,  right=of 1, minimum size=0.4cm, fill=black,  text=white, label=above:{$2$} ] (2) {};
            \node[draw, circle,  right=of 2, minimum size=0.4cm, fill=black,  text=white] (r) {};
            \node[above= 0.1cm of r] {....}; 
            \node[draw, circle,  right=of r, minimum size=0.4cm, fill=black,  text=white, label=above:{$n$}] (n) {};
            \node[draw, circle,  right=of n, minimum size=0.4cm, fill=black,  text=white, label=above:{$b$}] (b) {};

            \node[draw, circle,  above left=of r, minimum size=0.4cm, yshift=2cm, xshift=0.6cm, fill=black,  text=white, label=above:{$c$}] (c) {};

            \node[draw, circle,  below left=of r, minimum size=0.4cm, xshift=0.6cm, yshift=-1.5cm, fill=black,  text=white, label=below:{$d$}] (d) {};

            \draw[-, dashed] (a) -- (1);
            \draw[-, dashed] (1) -- (2);
            \draw[-, dashed] (2) -- (r);
            \draw[-, dashed] (r) -- (n);
            \draw[-, dashed] (n) -- (b);

            \draw[-] (d) -- (1);
            \draw[-] (d) -- (2);
            \draw[-] (d) -- (n);
            \draw[-] (d) -- (r);
            \draw[-] (d) -- (a);
            \draw[-] (d) -- (b);

            \draw[-] (a) -- (c);
            \draw[-] (c) -- (b);

        \end{tikzpicture}
        \label{g9n}
    }

\vspace{0.1cm}
   
  \subfloat[$H_{1}$]{
        \begin{tikzpicture}[scale=0.21, transform shape, every node/.style={font=\large}]
        \node[draw, fill=black, circle,minimum size=0.4cm,  text=white] (1) at (0, 0) {};
        \node[draw, fill=black, circle,minimum size=0.4cm,  text=white] (2) at (-6, 2.5) {};
        \node[draw, fill=black, circle,minimum size=0.4cm,  text=white] (3) at (2.2, -3) {};
        \node[draw, fill=black, circle,minimum size=0.4cm,  text=white] (4) at (-3.5, 1) {};
        \node[draw, fill=black, circle,minimum size=0.4cm,  text=white] (5) at (3.5, 1) {};
        \node[draw, fill=black, circle,minimum size=0.4cm,  text=white] (6) at (0, 3.5) {};
        \node[draw, fill=black, circle,minimum size=0.4cm,  text=white] (7) at (-2.2, -3) {};

        \draw[-] (1) -- (3);
        \draw[-] (1) -- (4);
        \draw[-] (1) -- (5);
        \draw[-] (1) -- (6);
        \draw[-] (1) -- (7);
        \draw[-] (2) -- (4);
        \draw[-] (2) -- (5);
        \draw[-] (2) -- (6);
        \draw[-] (2) -- (7);
        \draw[-] (3) -- (1);
        \draw[-] (3) -- (5);
        \draw[-] (3) -- (7);
        \draw[-] (4) -- (1);
        \draw[-] (4) -- (2);
        \draw[-] (4) -- (6);
        \draw[-] (4) -- (7);
        \draw[-] (5) -- (1);
        \draw[-] (5) -- (2);
        \draw[-] (5) -- (3);
        \draw[-] (5) -- (6);
        \draw[-] (5) -- (7);
        \draw[-] (6) -- (1);
        \draw[-] (6) -- (2);
        \draw[-] (6) -- (4);
        \draw[-] (6) -- (5);
        \draw[-] (7) -- (1);
        \draw[-] (7) -- (2);
        \draw[-] (7) -- (3);
        \draw[-] (7) -- (4);
        \draw[-] (7) -- (5);
        \end{tikzpicture}
        \label{$H_{1}$}
   }
    \hspace{0.17cm}
    \subfloat[$H_{2}$]{
        \begin{tikzpicture}[scale=0.20, transform shape, every node/.style={font=\large}]
        \node[draw, fill=black, circle,minimum size=0.4cm,  text=white] (7) at (0, 0.2) {};
        \node[draw, fill=black, circle,minimum size=0.4cm,  text=white] (4) at (0, -1.9) {};
        \node[draw, fill=black, circle,minimum size=0.4cm,  text=white] (2) at (3, -3.5) {};
        \node[draw, fill=black, circle,minimum size=0.4cm,  text=white] (3) at (-4, 1) {};
        \node[draw, fill=black, circle,minimum size=0.4cm,  text=white] (1) at (4, 1) {};
        \node[draw, fill=black, circle,minimum size=0.4cm,  text=white] (5) at (0, 3.5) {};
        \node[draw, fill=black, circle,minimum size=0.4cm,  text=white] (6) at (-3, -3.5) {};

        \draw[-] (7) -- (3);
        \draw[-] (7) -- (1);
        \draw[-] (7) -- (5);
        \draw[-] (7) -- (6);
        \draw[-] (7) -- (2);
        \draw[-] (2) -- (5);
        \draw[-] (2) -- (6);
        \draw[-] (3) -- (5);
        \draw[-] (3) -- (6);
        \draw[-] (4) -- (2);
        \draw[-] (4) -- (6);
        \end{tikzpicture}
        \label{$H_{2}$}
    }
    \hspace{0.17cm}
    \subfloat[$H_{3}$]{
        \begin{tikzpicture}[scale=0.20, transform shape, every node/.style={font=\large}]
        \node[draw, fill=black, circle,minimum size=0.4cm,  text=white] (1) at (0, 0) {};
        \node[draw, fill=black, circle,minimum size=0.4cm,  text=white] (2) at (-6, -0.5) {};
        \node[draw, fill=black, circle,minimum size=0.4cm,  text=white] (6) at (2.2, -3) {};
        \node[draw, fill=black, circle,minimum size=0.4cm,  text=white] (7) at (-3.5, 1) {};
        \node[draw, fill=black, circle,minimum size=0.4cm,  text=white] (5) at (3.5, 1) {};
        \node[draw, fill=black, circle,minimum size=0.4cm,  text=white] (3) at (0, 3.5) {};
        \node[draw, fill=black, circle,minimum size=0.4cm,  text=white] (4) at (-2.2, -3) {};

        \draw[-] (1) -- (3);
        \draw[-] (1) -- (4);
        \draw[-] (1) -- (5);
        \draw[-] (1) -- (6);
        \draw[-] (1) -- (7);
        \draw[-] (2) -- (4);
        \draw[-] (2) -- (6);
        \draw[-] (2) -- (7);
        \draw[-] (3) -- (1);
        \draw[-] (3) -- (5);
        \draw[-] (3) -- (7);
        \draw[-] (4) -- (1);
        \draw[-] (4) -- (2);
        \draw[-] (4) -- (6);
        \draw[-] (4) -- (7);
        \draw[-] (5) -- (1);
        \draw[-] (5) -- (3);
        \draw[-] (5) -- (6);
        \draw[-] (6) -- (1);
        \draw[-] (6) -- (2);
        \draw[-] (6) -- (4);
        \draw[-] (6) -- (5);
        \draw[-] (7) -- (1);
        \draw[-] (7) -- (2);
        \draw[-] (7) -- (3);
        \draw[-] (7) -- (4);
        \draw[-] (7) -- (6);
        \end{tikzpicture}
        \label{$H_{3}$}
    }
    \hspace{0.17cm}
    \subfloat[$H_{4}$]{
        \begin{tikzpicture}[scale=0.20, transform shape, every node/.style={font=\large}]
       
        
\node[draw, fill=black, circle,minimum size=0.4cm,  text=white] (6) at (0, 0) {};
\node[draw, fill=black, circle,minimum size=0.4cm,  text=white] (4) at (-6, 1) {};
\node[draw, fill=black, circle,minimum size=0.4cm,  text=white] (2) at (2.2, -3) {};
\node[draw, fill=black, circle,minimum size=0.4cm,  text=white] (7) at (-3.5, 1) {};
\node[draw, fill=black, circle,minimum size=0.4cm,  text=white] (5) at (3.5, 1) {};
\node[draw, fill=black, circle,minimum size=0.4cm,  text=white] (1) at (0, 3.5) {};
\node[draw, fill=black, circle,minimum size=0.4cm,  text=white] (3) at (-2.2, -3) {};

\draw[-] (6) -- (1);
\draw[-] (6) -- (2);
\draw[-] (6) -- (3);
\draw[-] (6) -- (5);
\draw[-] (6) -- (7);

\draw[-] (2) -- (5);
\draw[-] (2) -- (6);
\draw[-] (2) -- (7);

\draw[-] (3) -- (4);
\draw[-] (3) -- (6);
\draw[-] (3) -- (7);

\draw[-] (4) -- (1);
\draw[-] (4) -- (3);
\draw[-] (4) -- (7);

\draw[-] (5) -- (1);
\draw[-] (5) -- (2);
\draw[-] (5) -- (6);

\draw[-] (6) -- (1);
\draw[-] (6) -- (2);
\draw[-] (6) -- (3);
\draw[-] (6) -- (5);
\draw[-] (6) -- (7);

\draw[-] (7) -- (1);
\draw[-] (7) -- (2);
\draw[-] (7) -- (3);
\draw[-] (7) -- (4);
\draw[-] (7) -- (6);

        \end{tikzpicture}
        \label{$H_{4}$}
    }
  \hspace{0.17cm}
    \subfloat[$H_{5}$]{
        \begin{tikzpicture}[scale=0.20, transform shape, every node/.style={font=\large}]

\node[draw, fill=black, circle,minimum size=0.4cm,  text=white] (6) at (0, 0) {};
\node[draw, fill=black, circle,minimum size=0.4cm,  text=white] (4) at (0, -2) {};
\node[draw, fill=black, circle,minimum size=0.4cm,  text=white] (2) at (2.2, -3) {};
\node[draw, fill=black, circle,minimum size=0.4cm,  text=white] (7) at (-3.5, 1) {};
\node[draw, fill=black, circle,minimum size=0.4cm,  text=white] (5) at (3.5, 1) {};
\node[draw, fill=black, circle,minimum size=0.4cm,  text=white] (1) at (0, 3.5) {};
\node[draw, fill=black, circle,minimum size=0.4cm,  text=white] (3) at (-2.2, -3) {};

\draw[-] (6) -- (1);
\draw[-] (6) -- (2);
\draw[-] (6) -- (3);
\draw[-] (6) -- (7);
\draw[-] (7) -- (5);

\draw[-] (1) -- (3);
\draw[-] (2) -- (3);

\draw[-] (2) -- (5);
\draw[-] (2) -- (6);

\draw[-] (3) -- (4);
\draw[-] (3) -- (6);

\draw[-] (4) -- (2);
\draw[-] (4) -- (3);

\draw[-] (5) -- (1);
\draw[-] (5) -- (2);

\draw[-] (6) -- (1);
\draw[-] (6) -- (2);
\draw[-] (6) -- (3);
\draw[-] (6) -- (5);
\draw[-] (6) -- (7);

\draw[-] (7) -- (1);

\draw[-] (7) -- (6);

        \end{tikzpicture}
        \label{$H_{5}$}
    }
    
    \vspace{0.3cm}
    
    \subfloat[$H_{6}$]{
        \begin{tikzpicture}[scale=0.20, transform shape, every node/.style={font=\large}]
       
\node[draw, fill=black, circle,minimum size=0.4cm,  text=white] (6) at (0, 0) {};
\node[draw, fill=black, circle,minimum size=0.4cm,  text=white] (4) at (0, -2) {};
\node[draw, fill=black, circle,minimum size=0.4cm,  text=white] (2) at (2.2, -3) {};
\node[draw, fill=black, circle,minimum size=0.4cm,  text=white] (7) at (-3.5, 1) {};
\node[draw, fill=black, circle,minimum size=0.4cm,  text=white] (5) at (3.5, 1) {};
\node[draw, fill=black, circle,minimum size=0.4cm,  text=white] (1) at (0, 3.5) {};
\node[draw, fill=black, circle,minimum size=0.4cm,  text=white] (3) at (-2.2, -3) {};

\draw[-] (6) -- (1);
\draw[-] (6) -- (2);
\draw[-] (6) -- (3);
\draw[-] (6) -- (7);

\draw[-] (1) -- (3);
\draw[-] (2) -- (3);

\draw[-] (2) -- (5);
\draw[-] (2) -- (6);

\draw[-] (3) -- (4);
\draw[-] (3) -- (6);

\draw[-] (4) -- (2);
\draw[-] (4) -- (3);

\draw[-] (5) -- (1);
\draw[-] (5) -- (2);

\draw[-] (6) -- (1);
\draw[-] (6) -- (2);
\draw[-] (6) -- (3);
\draw[-] (6) -- (5);
\draw[-] (6) -- (7);

\draw[-] (7) -- (1);

\draw[-] (7) -- (6);

        \end{tikzpicture}
        \label{$H_{6}$}
    }
\hspace{0.15cm}
\subfloat[$H_{7}$]{
        \begin{tikzpicture}[scale=0.20, transform shape, every node/.style={font=\large}]

\node[draw, fill=black, circle,minimum size=0.4cm,  text=white] (2) at (0, -0.5) {};
\node[draw, fill=black, circle,minimum size=0.4cm,  text=white] (1) at (0, 2.5) {};
\node[draw, fill=black, circle,minimum size=0.4cm,  text=white] (3) at (3, 0) {};
\node[draw, fill=black, circle,minimum size=0.4cm,  text=white] (4) at (-3, -2.5) {};
\node[draw, fill=black, circle,minimum size=0.4cm,  text=white] (5) at (0, -2.5) {};
\node[draw, fill=black, circle,minimum size=0.4cm,  text=white] (6) at (3, -2.5) {};
\node[draw, fill=black, circle,minimum size=0.4cm,  text=white] (7) at (0, -4.3) {};

\draw[-] (2) -- (1);
\draw[-] (1) -- (3);
\draw[-] (6) -- (1);
\draw[-] (4) -- (1);

\draw[-] (2) -- (4);
\draw[-] (2) -- (6);

\draw[-] (2) -- (5);
\draw[-] (3) -- (6);

\draw[-] (5) -- (4);
\draw[-] (4) -- (7);

\draw[-] (5) -- (7);
\draw[-] (6) -- (7);

        \end{tikzpicture}
        \label{$H_{7}$}
    }
     \hspace{0.15cm}
    \subfloat[$H_{8}$]{
        \begin{tikzpicture}[scale=0.20, transform shape, every node/.style={font=\large}]
       
\node[draw, fill=black, circle,minimum size=0.4cm,  text=white] (2) at (0, -0.5) {};
\node[draw, fill=black, circle,minimum size=0.4cm,  text=white] (1) at (0, 2.5) {};
\node[draw, fill=black, circle,minimum size=0.4cm,  text=white] (3) at (3, 0) {};
\node[draw, fill=black, circle,minimum size=0.4cm,  text=white] (4) at (-3, -2.5) {};
\node[draw, fill=black, circle,minimum size=0.4cm,  text=white] (5) at (0, -2.5) {};
\node[draw, fill=black, circle,minimum size=0.4cm,  text=white] (6) at (3, -2.5) {};
\node[draw, fill=black, circle,minimum size=0.4cm,  text=white] (7) at (0, -4.3) {};

\draw[-] (2) -- (1);
\draw[-] (1) -- (3);
\draw[-] (6) -- (1);
\draw[-] (4) -- (1);

\draw[-] (2) -- (4);
\draw[-] (2) -- (6);

\draw[-] (2) -- (5);
\draw[-] (3) -- (6);

\draw[-] (5) -- (4);
\draw[-] (4) -- (7);

\draw[-] (6) -- (7);

\end{tikzpicture}
     \label{$H_{8}$}
}
\hspace{0.15cm}
\subfloat[$H_{9}$]{
        \begin{tikzpicture}[scale=0.20, transform shape, every node/.style={font=\large}]

\node[draw, fill=black, circle,minimum size=0.4cm,  text=white] (2) at (0, -0.5) {};
\node[draw, fill=black, circle,minimum size=0.4cm,  text=white] (1) at (0, 2.5) {};
\node[draw, fill=black, circle,minimum size=0.4cm,  text=white] (3) at (3, 0) {};
\node[draw, fill=black, circle,minimum size=0.4cm,  text=white] (4) at (-3, -2.5) {};
\node[draw, fill=black, circle,minimum size=0.4cm,  text=white] (5) at (0, -2.5) {};
\node[draw, fill=black, circle,minimum size=0.4cm,  text=white] (6) at (3, -2.5) {};
\node[draw, fill=black, circle,minimum size=0.4cm,  text=white] (7) at (0, -4.3) {};

\draw[-] (2) -- (1);
\draw[-] (1) -- (3);
\draw[-] (6) -- (1);
\draw[-] (4) -- (1);

\draw[-] (2) -- (4);
\draw[-] (2) -- (6);

\draw[-] (2) -- (5);
\draw[-] (3) -- (6);

\draw[-] (5) -- (4);
\draw[-] (4) -- (7);

     \end{tikzpicture}
         \label{$H_{9}$}
    }
   \hspace{0.15cm}
    \subfloat[$H_{10}$]{
        \begin{tikzpicture}[scale=0.20, transform shape, every node/.style={font=\large}]
       

\node[draw, fill=black, circle,minimum size=0.4cm,  text=white] (2) at (0, -0.5) {};
\node[draw, fill=black, circle,minimum size=0.4cm,  text=white] (1) at (0, 2.5) {};
\node[draw, fill=black, circle,minimum size=0.4cm,  text=white] (3) at (3, 0) {};
\node[draw, fill=black, circle,minimum size=0.4cm,  text=white] (4) at (-3, -2.5) {};
\node[draw, fill=black, circle,minimum size=0.4cm,  text=white] (5) at (0, -2.5) {};
\node[draw, fill=black, circle,minimum size=0.4cm,  text=white] (6) at (3, -2.5) {};
\node[draw, fill=black, circle,minimum size=0.4cm,  text=white] (7) at (0, -4.3) {};

\draw[-] (2) -- (1);
\draw[-] (1) -- (3);
\draw[-] (6) -- (1);
\draw[-] (4) -- (1);

\draw[-] (2) -- (4);
\draw[-] (2) -- (6);

\draw[-] (2) -- (5);
\draw[-] (3) -- (6);

\draw[-] (5) -- (4);
\draw[-] (4) -- (7);

\draw[-] (5) -- (7);

        \end{tikzpicture}
         \label{$H_{10}$}
    } 
     \hspace{0.15cm}
  \subfloat[$H_{11}$]{
 \begin{tikzpicture}[scale=0.20, transform shape, every node/.style={font=\large}]
\node[draw, fill=black, circle,minimum size=0.4cm,  text=white] (1) at (1.15, -2.5) {};

\node[draw, fill=black, circle,minimum size=0.4cm,  text=white] (2) at (2.3, 0.6) {};

\node[draw, fill=black, circle,minimum size=0.4cm,  text=white] (4) at (-3.7, 0.6) {};

\node[draw, fill=black, circle,minimum size=0.4cm,  text=white] (3) at (-0.7, 3) {};

\node[draw, fill=black, circle,minimum size=0.4cm,  text=white] (5) at (-2.95, -2.5) {};

\node[draw, fill=black, circle,minimum size=0.4cm,  text=white] (6) at (-5.2, -4.7) {};

\node[draw, fill=black, circle,minimum size=0.4cm,  text=white] (7) at (3.4, -4.7) {};

\draw[-] (1) -- (2);
\draw[-] (2) -- (3);
\draw[-] (3) -- (4);

\draw[-] (1) -- (5);
\draw[-] (1) -- (3);
\draw[-] (5) -- (4);

\draw[-] (5) -- (2);
\draw[-] (5) -- (3);
\draw[-] (1) -- (4);

\draw[-] (1) -- (6);
\draw[-] (4) -- (6);

\draw[-] (2) -- (7);
\draw[-] (5) -- (7);

 \end{tikzpicture}
 \label{$H_{11}$}
  }
    \caption{List of all minimal non-comparability graphs}
    \label{fig:min-non-comp-all}
\end{figure}
\begin{remark}
\label{remark-min-non-comp-fig}
In Figure~\ref{fig:min-non-comp-all}, several graph families are represented schematically. The classes $G_{n}^{5}$ and $G_{n}^{6}$ are depicted solely by their missing edges, indicated by dashed lines. $G_{n}^{5}$ denote the graph family which are complements of even cycles, and $G_{n}^{6}$ denote the graph family which are complements of odd cycles. For $G_{n}^{7}, G_{n}^{8},$ and $G_{n}^{9}$, the vertex set is partitioned into three horizontal layers:
\begin{itemize}
    \item For $G_{n}^{7}$ and $G_{n}^{8}$, the layers consist of the vertex sets $\{a, b, c, d\}$ (top), $\{1, 2, \dots, n\}$ (middle), and $\{x\}$ (bottom). 
    \item For $G_{n}^{9}$, the layers consist of $\{c\}$ (top), $\{a, 1, 2, \dots, n, b\}$ (middle), and $\{d\}$ (bottom).
\end{itemize}
Within each layer, dashed lines denote missing edges; specifically, the induced subgraph of each layer is a clique minus the set of dashed edges. All edges between vertices in different layers are drawn explicitly. While $H_{1}$ through $H_{11}$ are individual graphs, the $G_{n}^{i}$ notations represent infinite families of graphs parameterized by $n$.
\end{remark}

We now present the motivation for the problems studied in this paper.

\subsection{Motivation for the Problems}
\label{subsec-motivation}
Both word-representable graphs and comparability graphs are hereditary graph classes and, as such, they admit a characterization in terms of forbidden induced subgraphs. For comparability graphs, this characterization is known: the entire set of minimal non-comparability graphs has been identified by Gallai~\cite{Gallai1967TransitivOG}. In contrast, for word-representable graphs, the corresponding characterization remains open; that is, a complete description of the minimal non-word-representable graphs is not yet known. Determining this class, therefore, constitutes a natural and significant problem.

Since comparability graphs form a subclass of word-representable graphs, it is natural to investigate the relationship between their respective minimal forbidden induced subgraphs. Indeed, because every proper induced subgraph of a minimal non-comparability graph is a comparability graph (and thus inherently word-representable), it follows logically that any minimal non-comparability graph must be either word-representable or minimal non-word-representable itself. This structural classification is illustrated via the Venn diagram in Figure~\ref{fig:graph_classes_venn}. Consequently, a fundamental question arises: do there exist graphs that are simultaneously minimal non-comparability and minimal non-word-representable, and if so, how can we identify them? More generally, common minimal forbidden induced subgraphs provide insight into the structural similarities and differences between graph classes. While such comparisons are inherently qualitative, it is also natural to seek quantitative measures of the relationship between the two classes. This motivates the study of covering parameters, in particular, the minimum number of comparability graphs whose union covers the edges of a given graph.  A precise definition is given in Section~\ref{covering-section}.

The problem we study, posed in~\cite{kenkireth2026word}, asks for the cover number of a word-representable graph by comparability graphs. For any word-representable graph $G$, its cover number by comparability graphs can be viewed as a measure of how far $G$ is from being a comparability graph; larger values indicate greater structural deviations. This parameter complements the study of minimal forbidden induced subgraphs by providing a quantitative perspective on the relationship between the two classes. Hence, we study these two problems in this paper.

\begin{figure}[htbp]
    \centering
    \resizebox{0.75\textwidth}{!}{%
        \begin{tikzpicture}[thick, font=\sffamily\small]
        
            \draw[black, fill=white] (-1, -1) rectangle (11.5, 6);
            \node[anchor=north west] at (-0.8, 5.7) {\textbf{All Graphs}};
        
            \draw[black, fill=gray!5] (3, 3) circle (2.6cm);
            \node[black] at (3, 4.35) {\textbf{Word-Representable}};
        
            \draw[black, fill=gray!20] (2, 2.5) circle (1.3cm);
            \node[black] at (1.8, 2.5) {\textbf{Comparability}};
        
            \draw[black, fill=gray!12] (5.3, 2.5) circle (1.3cm);
            
            \begin{scope}
                \clip (3, 3) circle (2.6cm) (-1, -1) rectangle (11.5, 7);
                
                \path[pattern={Lines[angle=45, distance=5pt, line width=0.6pt]}, pattern color=black] 
                    (5.3, 2.5) circle (1.3cm);
            \end{scope}
        
            \draw[black] (3, 3) circle (2.6cm);
            \draw[black] (5.3, 2.5) circle (1.3cm);
        
            \node[black, align=center] at (5.1, 2.5) {\textbf{Minimal Non-}\\\textbf{Comparability}};
        
            \node[black, align=center, anchor=west] at (6.15, 4.5) {\textbf{Minimal Non-Comparability \&}\\\textbf{Minimal Non-Word-Representable}\\\textbf{(Intersection Class)}};
            
            \draw[->, black, line width=1.2pt] (7.2, 4.0) to[out=210, in=45] (6.0, 3.1);
        
        \end{tikzpicture}%
    }
    \caption{Venn diagram depicting relations between graph classes: The non-word-representable region of the minimal non-comparability class forms exactly the intersection class of minimal non-comparability and minimal non-word-representable graphs.}
    \label{fig:graph_classes_venn}
\end{figure}

\section{Classification of Minimal Non-Comparability Graphs with Respect to Word-Representability}
\label{section-classification}
Figure~\ref{fig:min-non-comp-all} presents the complete set of minimal non-comparability graphs. 
This collection consists of two types: individual graphs \(H_i\) and infinite families of graphs \(G_n^i\) parameterized by \(n\). 
For each family \(G_n^i\), the parameter \(n\) is at least the lower bound indicated in the figure, which corresponds to the smallest instance in the family that is non-comparability.

In this section, we classify minimal non-comparability graphs according to their word-representability. 
The results are summarized in Theorems~\ref{wr-graphs-thm} and~\ref{non-wr-graphs-thm}. 
Theorem~\ref{wr-graphs-thm} identifies the word-representable graphs among the collection, while Theorem~\ref{non-wr-graphs-thm} identifies the non-word-representable ones, which form the intersection of minimal non-comparability graphs with minimal non-word-representable graphs.

\begin{theorem}
\label{wr-graphs-thm}
The minimal non-comparability graphs shown in Figure~\ref{fig:min-non-comp-all} that are word-representable are precisely the following:
\begin{enumerate}
    \item All graphs in the families \( G_{n}^{1}, G_{n}^{2}, G_{n}^{3}, G_{n}^{5}, G_{n}^{6}, G_{n}^{7}, G_{n}^{8} \).
    \item The graphs \( H_{2}, H_{3}, \dots, H_{11} \), together with \( G_{2}^{9} \).
\end{enumerate}
\end{theorem}

\begin{theorem}
\label{non-wr-graphs-thm}
The minimal non-comparability graphs shown in Figure~\ref{fig:min-non-comp-all} that are not word-representable are precisely the following:
\begin{enumerate}
    \item All graphs in the families \( G_{n}^{4} \) and \( G_{n}^{9} \) for \( n \ge 3 \).
    \item The graph \( H_{1} \).
\end{enumerate}
Equivalently, these are exactly the minimal non-comparability graphs that are also minimal non-word-representable graphs.
\end{theorem}

The proofs of Theorems~\ref{wr-graphs-thm} and~\ref{non-wr-graphs-thm} are completed in Section~3. Subsection~\ref{section-min-semi-trans} establishes that the graphs listed in Theorem~\ref{wr-graphs-thm} are word-representable, while Subsection~\ref{section-minimal-non-semi-trans} shows that the graphs listed in Theorem~\ref{non-wr-graphs-thm} are not word-representable. Together, these two subsections yield a complete classification of the minimal non-comparability graphs with respect to word-representability.

\begin{remark}
In a preliminary version of this work~\cite{kenkireth2026word}, the family \( G_{n}^{9} \) was stated to be non-word-representable for all \( n \ge 2 \). We clarify that the case \( n=2 \) is, in fact, word-representable. Accordingly, \( G_{2}^{9} \) is included among the word-representable graphs, while \( G_{n}^{9} \) for \( n \ge 3 \) are non-word-representable.
\end{remark}

\subsection{Minimal Non-Comparability Graphs that are Word-Representable}
\label{section-min-semi-trans}

In this subsection, we establish the word-representability of all graphs listed in Theorem~\ref{wr-graphs-thm} via a sequence of intermediate propositions.

\begin{proposition}
\label{3-color-graphs-prop}
Every graph in the families \( G_{n}^{1} \), \( G_{n}^{2} \), and \( G_{n}^{3} \), as shown in Figure~\ref{fig:min-non-comp-all}, is word-representable.
\end{proposition}
\begin{proof}
Consider an arbitrary graph \( G \in G_{n}^{1} \), where \( n \geq 2 \). The vertex set of \( G \) admits a partition into three independent sets \( A \), \( B \), and \( C \), where
\( A = \{1,3,\ldots,2n-1\} \), \( B = \{2,4,\ldots,2n\} \), and \( C = \{2n+1\} \). Similarly, for any graph \( G \in G_{n}^{2} \), where \( n \geq 2 \), the vertex set admits a partition into three independent sets \( A \), \( B \), and \( C \), where
\( A = \{1, x, y\} \), \( B = \{2,4,6,\ldots,2n\} \), and \( C = \{3,5,7,\ldots,2n+1\} \). Likewise, for any graph \( G \in G_{n}^{3} \), where \( n \geq 3 \), the vertex set admits a partition into three independent sets \( A \), \( B \), and \( C \), where
\( A = \{1,2n\} \), \( B = \{2,4,6,\ldots,2n-2,x\} \), and \( C = \{3,5,7,\ldots,2n-1,y\} \).

Therefore, every graph in these families is 3-colorable. Figure~\ref{fig:3-colo-graphs} illustrates the 3-colorability of these graph classes.
\begin{figure}[htbp]
    \centering
    \subfloat[$G_{n}^{1}$ $;$ $n \geq 2$]{
\begin{tikzpicture}[scale=0.37, transform shape, every node/.style={font=\Huge}]

\node[draw, circle, fill=red, text=white, minimum size=0.4cm, label=right:{$1$}] (1) at (2, 0) {};
\node[draw, circle, fill=blue, text=white, minimum size=0.4cm, label=right:{$2$}] (2) at (3.5, 1.65) {};
\node[draw, circle, fill=red, text=white, minimum size=0.4cm, label=right:{$3$}] (3) at (2.9, 3.7) {};
\node[draw, circle, fill=white, minimum size=0.4cm] (r) at (0.5, 5.5) {};
\node[draw, circle, fill=white, minimum size=0.4cm] (r1) at (-1.9, 3.7) {};
\node[draw, circle, fill=blue, text=white, minimum size=0.4cm, label=left:{$2n$}] (2n) at (-2.5, 1.65) {};
\node[draw, circle, fill=brown, text=white, minimum size=0.4cm, label=left:{$2n+1$}] (2n+1) at (-1, 0) {};

\draw[-] (1) -- (2);
\draw[-] (2) -- (3);
\draw[-] (3) -- (r) node[midway, sloped, above, yshift=5pt] {. . . . . . .}; 
\draw[-] (r) -- (r1);
\draw[-] (r1) -- (2n);
\draw[-] (2n) -- (2n+1);
\draw[-] (2n+1) -- (1);

\end{tikzpicture}

\label{fig:3-colo-g1n}
}
\hfill
\subfloat[$G_{n}^{2}$ $;$ $n \geq 2$]{
 
\begin{tikzpicture}[scale=0.37, transform shape, every node/.style={font=\Huge}]

\node[draw, circle, fill=red, text=white, minimum size=0.4cm, label=below:{$1$}] (1) at (0, -2) {};
\node[draw, circle, fill=blue, text=white, minimum size=0.4cm, label=right:{$2$}] (2) at (2.7, 0) {};
\node[draw, circle, fill=brown, text=white, minimum size=0.4cm, label=left:{$2n+1$}] (2n+1) at (-2.7, 0) {};
\node[draw, circle, fill=brown, text=white, minimum size=0.4cm, label=right:{$3$}] (3) at (3, 2.3) {};
\node[draw, circle, fill=blue, text=white, minimum size=0.4cm, label=right:{$4$}] (4) at (1.5, 4) {};
\node[draw, circle, fill=white, minimum size=0.4cm] (r) at (-1.5, 4) {};
\node[draw, circle, fill=blue, text=white, minimum size=0.4cm, label=left:{$2n$}] (2n) at (-3, 2.3) {};
\node[draw, circle, fill=red, text=white, minimum size=0.4cm, label=left:{$x$}] (x) at (-4, -2) {};
\node[draw, circle, fill=red, text=white, minimum size=0.4cm, label=right:{$y$}] (y) at (4, -2) {};

\draw[-] (1) -- (2);
\draw[-] (2) -- (3);
\draw[-] (3) -- (4);
\draw[-] (4) -- (r) node[midway, above, yshift=5pt] {. . . . . . .};
\draw[-] (r) -- (2n);
\draw[-] (2n) -- (2n+1);
\draw[-] (2n+1) -- (1);
\draw[-] (1) -- (3);
\draw[-] (1) -- (4);
\draw[-] (1) -- (r);
\draw[-] (1) -- (2n);
\draw[-] (1) -- (2n+1);
\draw[-] (2n+1) -- (x);
\draw[-] (2) -- (y);

\end{tikzpicture}

\label{fig:3-colo-g2n}
}
      \hfill
\subfloat[$G_{n}^{3}$ $;$ $n \geq 3$]{
     
\begin{tikzpicture}[scale=0.37, transform shape, every node/.style={font=\Huge}]

\node[draw, circle, fill=red, text=white, minimum size=0.4cm, label=below:{$1$}] (1) at (1.15, -2.5) {};
\node[draw, circle, fill=blue, text=white, minimum size=0.4cm, label=right:{$2$}] (2) at (3.4, 0) {};
\node[draw, circle, fill=brown, text=white, minimum size=0.4cm, label=left:{$2n-1$}] (2n-1) at (-3.9, 0) {};
\node[draw, circle, fill=brown, text=white, minimum size=0.4cm, label=right:{$3$}] (3) at (3.37, 2.3) {};
\node[draw, circle, fill=white, minimum size=0.4cm] (r1) at (-3.9, 2.3) {};
\node[draw, circle, fill=blue, text=white, minimum size=0.4cm, label=right:{$4$}] (4) at (1.2, 4) {};
\node[draw, circle, fill=white, minimum size=0.4cm] (r) at (-1.7, 4) {};
\node[draw, circle, fill=red, text=white, minimum size=0.4cm, label=below:{$2n$}] (2n) at (-1.65, -2.5) {};
\node[draw, circle, fill=blue, text=white, minimum size=0.4cm, label=left:{$x$}] (x) at (-4.5, -4.9) {};
\node[draw, circle, fill=brown, text=white,minimum size=0.4cm, label=right:{$y$}] (y) at (4, -4.9) {};

\draw[-] (1) -- (2);
\draw[-] (2) -- (3);
\draw[-] (3) -- (4);
\draw[-] (4) -- (r) node[midway, above, yshift=5pt] {. . . . . . .};
\draw[-] (1) -- (3);
\draw[-] (1) -- (4);
\draw[-] (1) -- (r);
\draw[-] (1) -- (r1);
\draw[-] (1) -- (2n-1);
\draw[-] (2n) -- (2);
\draw[-] (2n) -- (3);
\draw[-] (2n) -- (4);
\draw[-] (2n) -- (r);
\draw[-] (2n) -- (r1);
\draw[-] (2n) -- (2n-1);
\draw[-] (r) -- (r1);
\draw[-] (2n-1) -- (r1);
\draw[-] (2n-1) -- (x);
\draw[-] (2) -- (y);
\draw[-] (2n) -- (y);
\draw[-] (1) -- (x);

\end{tikzpicture}

\label{fig:3-colo-g3n}
}
\caption{$3$-colorable minimal non-comparability graph classes}
\label{fig:3-colo-graphs}
\end{figure}
By Theorem~\ref{3-col-semi}, every 3-colorable graph is semi-transitive. Hence, all graphs in these families are word-representable. \qed
\end{proof}

\begin{proposition}
\label{gn5-semi-thm}
Every graph in the family \(G_{n}^{5}\), as shown in Figure~\ref{fig:min-non-comp-all}, is word-representable.
\end{proposition}

\begin{proof}
Consider an arbitrary graph \(G \in G_{n}^{5}\), with \(n \ge 3\), where the vertex set is partitioned as \(V(G) = A \cup B\), with \(A = \{a_1, \dots, a_n\}\) and \(B = \{b_1, \dots, b_n\}\). Both \(A\) and \(B\) induce cliques. The missing \(2n\) edges form a cycle alternating between \(A\) and \(B\), explicitly defined as:
\begin{itemize}
    \item \(\{a_i, b_i\}\) for all \(1 \le i \le n\),
    \item \(\{a_i, b_{i-1}\}\) for all \(2 \le i \le n\), and the edge \(\{a_1, b_n\}\mathrm{.}\)
\end{itemize}

To prove word-representability, we exhibit a semi-transitive orientation (Theorem~\ref{wrg=semi}). We define a strict total ordering \(\prec\) on the entire vertex set by interleaving the elements of \(A\) and \(B\) as follows:
\[ a_1 \prec b_1 \prec a_2 \prec b_2 \prec \dots \prec a_n \prec b_n \]
Orient each existing edge of \(G\) from the vertex that is smaller under \(\prec\) to the vertex that is larger. This orientation is by construction a travel from left to right along the linear order, and is clearly acyclic.

Suppose for a contradiction that there exists a directed path \(P = u_1 \to \dots \to u_k\) violating semi-transitivity, meaning that \(u_1 \to u_k\) is a directed edge in \(G\) and some non-consecutive pair \(\{u_i, u_j\} \notin E(G)\) for \(1 \le i < j \le k\) with \(j > i + 1\). Because all edges are oriented according to the total ordering \(\prec\), the vertices in \(P\) must appear in strictly increasing order with respect to \(\prec\).

We now check the missing edges of \(G\) to see if they can serve as the violating non-consecutive pair \(\{u_i, u_j\}\):
\begin{itemize}
    \item Type 1 (\(\{a_m, b_m\}\) for \(1 \le m \le n\)): In the defined total ordering \(\prec\), the vertex \(b_m\) is the immediate successor of \(a_m\). There is no vertex \(v \in V(G)\) such that \(a_m \prec v \prec b_m\). Because any directed path must move in strictly increasing order under \(\prec\), if both vertices \(a_m\) and \(b_m\) appear on \(P\), they must appear as immediately consecutive elements on the path (i.e., \(b_m\) immediately follows \(a_m\)). This directly contradicts the semi-transitivity violation requirement that \(\{u_i, u_j\}\) must be a non-consecutive pair (\(j > i + 1\)).
    \item Type 2 (\(\{a_m, b_{m-1}\}\) for \(2 \le m \le n\)): In the total ordering \(\prec\), the vertex \(a_m\) is the immediate successor of \(b_{m-1}\). There is no vertex \(v \in V(G)\) such that \(b_{m-1} \prec v \prec a_m\). By the exact same logic as Type 1, if both vertices appear on \(P\), they must be immediately consecutive elements on the path, which directly contradicts the assumption that they form a non-consecutive pair.
   \item Type 3 (\(\{a_1, b_n\}\)): Under \(\prec\), the vertex \(a_1\) is the absolute minimum element and \(b_n\) is the absolute maximum element of the entire vertex set. Any directed path containing both vertices must therefore begin at \(a_1\) and end at \(b_n\), forcing \(u_1 = a_1\) and \(u_k = b_n\). Semi-transitivity requires the shortcut selection \(\{u_1, u_k\}\) to be an active directed edge in \(G\). However, \(\{a_1, b_n\}\) is a missing edge by definition, yielding a direct contradiction.

\end{itemize}
Hence, no such path \(P\) can exist. The orientation is semi-transitive, and \(G\) is word-representable.
\end{proof}

\begin{proposition}
\label{gn6-wr-prop}
Every graph in the family \(G_{n}^{6}\), as shown in Figure~\ref{fig:min-non-comp-all}\@, is word-representable.
\end{proposition}

\begin{proof}
Consider an arbitrary graph \(G \in G_n^6\), with \(n \ge 3\), and vertex set \(V(G) = A \cup B \cup \{c_0\}\), where \(A = \{a_1, \dots, a_n\}\) and \(B = \{b_1, \dots, b_n\}\). The sets \(A\) and \(B\) each induce a clique in \(G\). The \(2n+1\) missing edges form a single cycle spanning all vertices, alternating between \(A\) and \(B\) and closing through \(c_0\), explicitly defined as:
\begin{itemize}
    \item \(\{a_i, b_i\}\) for all \(1 \le i \le n\),
    \item \(\{a_i, b_{i-1}\}\) for all \(2 \le i \le n\),
    \item \(\{c_0, a_1\}\), and the edge \(\{c_0, b_n\}\).
\end{itemize}

To prove word-representability, we exhibit a semi-transitive orientation (Theorem~\ref{wrg=semi}). We define a strict total ordering \(\prec\) on the entire vertex set by placing \(c_0\) at the minimum position and interleaving the elements of \(A\) and \(B\) as follows:
\[ c_0 \prec a_1 \prec b_1 \prec a_2 \prec b_2 \prec \dots \prec a_n \prec b_n \]
We orient each edge of \(G\) from the smaller vertex to the larger vertex under \(\prec\). Because this orientation is consistent with a strict total order, it is inherently acyclic.

Suppose for a contradiction that there exists a directed path \(P = u_1 \to \dots \to u_k\) violating semi-transitivity, meaning that \(u_1 \to u_k\) is a directed edge in \(G\) and some non-consecutive pair \(\{u_i, u_j\} \notin E(G)\) for \(1 \le i < j \le k\) with \(j > i + 1\). Because all edges are oriented according to the total ordering \(\prec\), the vertices in \(P\) must appear in strictly increasing order with respect to \(\prec\).

We now check the missing edges of \(G\) to see if they can serve as the violating non-consecutive pair \(\{u_i, u_j\}\):
\begin{itemize}
    \item Type 1 (\(\{a_m, b_m\}\) for \(1 \le m \le n\)): In the defined total ordering \(\prec\), the vertex \(b_m\) is the immediate successor of \(a_m\). There is no vertex \(v \in V(G)\) such that \(a_m \prec v \prec b_m\). Because any directed path must move in strictly increasing order under \(\prec\), if both vertices \(a_m\) and \(b_m\) appear on \(P\), they must appear as immediately consecutive elements on the path. This directly contradicts the semi-transitivity violation requirement that \(\{u_i, u_j\}\) must be a non-consecutive pair (\(j > i + 1\)).
    \item Type 2 (\(\{a_m, b_{m-1}\}\) for \(2 \le m \le n\)): In the total ordering \(\prec\), the vertex \(a_m\) is the immediate successor of \(b_{m-1}\). There is no vertex \(v \in V(G)\) such that \(b_{m-1} \prec v \prec a_m\). By the exact same logic as Type 1, if both vertices appear on \(P\), they must be immediately consecutive elements on the path, which directly contradicts the assumption that they form a non-consecutive pair.
    \item Type 3 (\(\{c_0, a_1\}\)): In the total ordering \(\prec\), the vertex \(a_1\) is the immediate successor of \(c_0\). There is no vertex \(v \in V(G)\) such that \(c_0 \prec v \prec a_1\). By the exact same successor logic, if both vertices appear on \(P\), they must be immediately consecutive elements on the path, which directly contradicts the non-consecutive assumption.
    \item Type 4 (\(\{c_0, b_n\}\)): Under \(\prec\), the vertex \(c_0\) is the minimum element and \(b_n\) is the maximum element of the entire vertex set. Any directed path containing both vertices must therefore begin at \(c_0\) and end at \(b_n\), forcing \(u_1 = c_0\) and \(u_k = b_n\). Semi-transitivity requires the shortcut selection \(\{u_1, u_k\}\) to be an active directed edge in \(G\). However, \(\{c_0, b_n\}\) is a missing edge by definition, yielding a direct contradiction.
\end{itemize}
Hence, no such path \(P\) can exist. The orientation is semi-transitive, and \(G\) is word-representable.
\end{proof}

\begin{proposition}
\label{theorem-g7n-g8n}
Each graph in the families \(G_{n}^{7}\) and \(G_{n}^{8}\), as shown in Figure~\ref{fig:min-non-comp-all}, is word-representable.
\end{proposition}

\begin{proof}
Consider an arbitrary graph \(G \in \mathcal{G}_n^7\), where \(n \ge 1\), with vertex partition$
V(G) = V_1 \cup V_2 \cup V_3,
$
where \(V_1 = \{a,b,c,d\}\), \(V_2 = \{1,2,\dots,n\}\), and \(V_3 = \{x\}\). To prove that \(G\) is word-representable, we show that it admits a semi-transitive orientation (Theorem~\ref{wrg=semi}). We orient the edges of \(G\) as follows:
\begin{itemize}
    \item Edges with both endpoints in \(V_1\): \(b \to d\), \(d \to a\), \(a \to c\);
    \item Edges with both endpoints in \(V_2\): \(i \to j\) whenever \(i < j\);
\item Edges between distinct parts: all edges are oriented from \(V_i\) to \(V_j\) whenever \(i < j\).
\end{itemize}

We first show that the above orientation is acyclic. Assign to each vertex \(v \in V\) an index \(i \in \{1,2,3\}\) such that \(v \in V_i\). By construction, every edge between distinct parts is oriented from \(V_i\) to \(V_j\) whenever \(i < j\). Consequently, along any directed path, the associated indices form a non-decreasing sequence, and the index strictly increases whenever the path traverses an edge between different parts. It follows that no directed path can return to a vertex in a lower-indexed part. In particular, a directed cycle cannot contain vertices from more than one part. Therefore, any directed cycle, if it exists, must be entirely contained within a single part \(V_i\), that is, all its edges have both endpoints in \(V_i\). It remains to verify that no directed cycle is contained within any part \(V_i\). In \(G[V_1]\), the induced subgraph contains only the edges \(b \to d\), \(d \to a\), and \(a \to c\), which clearly form no directed cycle.  The subgraph \(G[V_2]\) is oriented according to the natural order: \(i \to j\) whenever \(i < j\). Hence, no directed cycle can exist in \(G[V_2]\). Finally, \(G[V_3]\) consists of the single vertex \(x\), and therefore contains no edges and no cycles. Therefore, the orientation defined for \(G\) is acyclic.

We now show that the defined orientation is semi-transitive by analyzing the structure of directed paths in \(G\). We distinguish the following cases:

\begin{itemize}
    \item Directed paths entirely contained within a single induced subgraph \(G[V_i]\), for \(i \in \{1,2,3\}\).
    
    \item Directed paths that traverse more than one part. Since all edges between distinct parts are oriented from \(V_i\) to \(V_j\) whenever \(i < j\), any such directed path can only move from \(V_1\) to \(V_2\), and from \(V_2\) to \(V_3\), and never in the reverse direction. Consequently, every such path is of one of the following forms:
    \begin{itemize}
        \item a path from \(V_1\) to \(V_2\);
        \item a path from \(V_2\) to \(V_3\);
        \item a path from \(V_1\) through \(V_2\) to \(V_3\).
    \end{itemize}
\end{itemize}

Any directed path with both endpoints in $V_{2}$ preserves semi-transitivity. The only missing edges in $G[V_{2}]$ are consecutive pairs \(\{i,i+1\}\), hence no semi-transitivity-violating configuration can occur. Any directed path with both endpoints in $V_{1}$ preserves semi-transitivity.  The longest directed path is \(b \to d \to a \to c\), and the edge \(b \to c\) is absent. Any other path has length at most $2$. Hence no semi-transitivity violation arises.

We now consider directed paths that start in \(V_1\) and end in \(V_2\). In the induced subgraph \(G[V_1]\), only the vertices \(a\) and \(d\) are adjacent to vertices in \(V_2\). More precisely, \(a\) is adjacent to all vertices of \(V_2\) except \(1\), while \(d\) is adjacent to all vertices of \(V_2\) except \(n\).

First, consider directed paths that start at \(a\). Since all edges between \(V_1\) and \(V_2\) are oriented from \(V_1\) to \(V_2\), any such path proceeds from \(a\) directly into \(V_2\) and then continues entirely within \(G[V_2]\). By construction, \(a\) is adjacent to every possible successor in \(V_2\) along such a path, except for the vertex \(1\), which cannot appear after \(a\) in a directed path due to the orientation \(i \to j\) whenever \(i < j\) inside $G[V_{2}]$. Moreover, any directed path contained in \(G[V_2]\) preserves semi-transitive. Hence, no shortcut violating semi-transitivity can arise, and every such path preserves semi-transitivity.

Next, consider directed paths that start at \(d\). Such paths either proceed directly from \(d\) to a vertex in \(V_2\), or pass through \(a\) before entering \(V_2\). In the latter case, the subpath starting at \(a\) satisfies semi-transitivity by the argument above. Furthermore, \(d\) is adjacent to all vertices of \(V_2\) except \(n\), and the vertex \(n\) can only appear as the terminal vertex of a directed path in \(G[V_2]\). Consequently, no such paths can violate semi-transitivity. A similar argument applies when the path proceeds directly from \(d\) into \(V_2\).
Therefore, all directed paths starting in \(V_1\) and ending in \(V_2\) preserve semi-transitivity.

We now consider directed paths that terminate at the vertex \(x \in V_3\). The vertex \(x\) is adjacent to all vertices of \(G\) except \(b\) and \(c\). Observe that the vertices \(b\) and \(c\) have no neighbors in \(V_2 \cup V_3\), and hence they do not appear in any directed path that traverses multiple parts and reaches \(x\). As established earlier, any directed path that traverses multiple parts can only move from \(V_i\) to \(V_j\) with \(i \le j\), and never in the reverse direction. Extending such a path by the final vertex \(x\) yields a directed path terminating at \(x\). Since \(x\) is adjacent to all vertices that can appear in such a path, no shortcut violating semi-transitivity can arise from this extension. Moreover, we have already shown that all such paths prior to the addition of \(x\) preserve semi-transitivity.

Therefore, every directed path terminating at \(x\) preserves semi-transitivity. Consequently, the defined orientation of \(G\) is semi-transitive.

The same argument applies to graphs in the family \(G_n^8\). The only difference in the orientation lies in the induced subgraph \(G[V_1]\), where an additional edge between \(b\) and \(c\) is present, oriented as \(c \to b\).  In this case, the orientation within \(G[V_1]\) is given by \(d \to a\), \(c \to a\), \(c \to b\), and \(d \to b\). Thus any directed path within \(G[V_1]\) preserves semi-transitivity. All other edges are oriented according to the same rules as for the family \(G_n^7\). In particular, the orientations in \(G[V_2]\), \(G[V_3]\), and all edges between different parts remain unchanged. Since the only two vertices in $V_{1}$ connecting to other parts are $d$ and $a$, and the orientation between them is unchanged here, (that is $d \rightarrow a$) the modifications do not affect the arguments for acyclicity or semi-transitivity, the same reasoning for \(G_n^7\) applies here as well. Hence, every graph in \(G_n^8\) is also word-representable.
\end{proof}

\begin{remark}
The graphs \(H_{2}, H_{3}, \dots, H_{11}\), together with \(G_{2}^{9}\), are word-representable. For each of these graphs, a semi-transitive orientation is given in Figure~\ref{fig:min-non-comp-part4-semi-trans}.
\end{remark}

\begin{figure}[htbp]
    \centering

\subfloat[$H_{2}$]{

\begin{tikzpicture}[
    scale=0.35,
    transform shape,
    every node/.style={font=\Huge}
]

\tikzset{
    vertex/.style={
        draw,
        circle,
        fill=black,
        minimum size=0.4cm,
        inner sep=0pt
    },
    edge/.style={
        thick,
        postaction={decorate},
        decoration={
            markings,
            mark=at position 0.5 with {\arrow{>}}
        }
    }
}

\node[vertex, label=above left:$1$] (7) at (0, 0.2) {};
\node[vertex, label=above:$2$] (4) at (0, -1.9) {};
\node[vertex, label=right:$3$] (2) at (3, -3.5) {};
\node[vertex, label=left:$4$] (3) at (-4, 1) {};
\node[vertex, label=right:$5$] (1) at (4, 1) {};
\node[vertex, label=above:$6$] (5) at (0, 3.5) {};
\node[vertex, label=left:$7$] (6) at (-3, -3.5) {};

\draw[edge] (7) -- (3);
\draw[edge] (7) -- (1);
\draw[edge] (5) -- (7);
\draw[edge] (7) -- (6);
\draw[edge] (7) -- (2);
\draw[edge] (5) -- (2);
\draw[edge] (2) -- (6);
\draw[edge] (5) -- (3);
\draw[edge] (3) -- (6);
\draw[edge] (4) -- (2);
\draw[edge] (4) -- (6);

\end{tikzpicture}
}
\hspace{0.22cm}
\subfloat[$H_{3}$]{
\begin{tikzpicture}[
    scale=0.35,
    transform shape,
    every node/.style={font=\Huge}
]

\tikzset{
    vertex/.style={
        draw,
        circle,
        fill=black,
        minimum size=0.4cm,
        inner sep=0pt
    },
    edge/.style={
        thick,
        postaction={decorate},
        decoration={
            markings,
            mark=at position 0.5 with {\arrow{>}}
        }
    }
}

\node[vertex, label=above left:$1$] (1) at (0, 0) {};
\node[vertex, label=left:$2$] (2) at (-6, -0.5) {};
\node[vertex, label=right:$3$] (6) at (2.2, -3) {};
\node[vertex, label=left:$4$] (7) at (-3.5, 1) {};
\node[vertex, label=right:$5$] (5) at (3.5, 1) {};
\node[vertex, label=above:$6$] (3) at (0, 3.5) {};
\node[vertex, label=left:$7$] (4) at (-2.2, -3) {};

\draw[edge] (1) -- (3);
\draw[edge] (1) -- (4);
\draw[edge] (5) -- (1);
\draw[edge] (1) -- (6);
\draw[edge] (1) -- (7);

\draw[edge] (5) -- (3);
\draw[edge] (3) -- (7);

\draw[edge] (4) -- (2);

\draw[edge] (5) -- (6);

\draw[edge] (6) -- (2);
\draw[edge] (6) -- (4);

\draw[edge] (7) -- (2);

\draw[edge] (4) -- (7);
\draw[edge] (6) -- (7);

\end{tikzpicture}
}
     \hspace{0.22cm}
    \subfloat[$H_{4}$]{
       \begin{tikzpicture}[
    scale=0.35,
    transform shape,
    every node/.style={font=\Huge}
]

\tikzset{
    vertex/.style={
        draw,
        circle,
        fill=black,
        minimum size=0.4cm,
        inner sep=0pt
    },
    edge/.style={
        thick,
        postaction={decorate},
        decoration={
            markings,
            mark=at position 0.5 with {\arrow{>}}
        }
    }
}
        
\node[draw, circle, fill=black, text=white, minimum size=0.4cm,
      label=above:$1$] (1) at (0, 3.5) {};

\node[draw, circle, fill=black, text=white, minimum size=0.4cm,
      label=right:$2$] (2) at (2.2, -3) {};

\node[draw, circle, fill=black, text=white, minimum size=0.4cm,
      label=left:$3$] (3) at (-2.2, -3) {};

\node[draw, circle, fill=black, text=white, minimum size=0.4cm,
      label=left:$4$] (4) at (-6, 1) {};

\node[draw, circle, fill=black, text=white, minimum size=0.4cm,
      label=right:$5$] (5) at (3.5, 1) {};

\node[draw, circle, fill=black, text=white, minimum size=0.4cm,
      label=above left:$6$] (6) at (0, 0) {};

\node[draw, circle, fill=black, text=white, minimum size=0.4cm,
      label=below left:$7$] (7) at (-3.5, 1) {};

\draw[edge] (6) -- (1);
\draw[edge] (6) -- (2);
\draw[edge] (6) -- (3);
\draw[edge] (6) -- (5);
\draw[edge] (6) -- (7);

\draw[edge] (2) -- (7);

\draw[edge] (3) -- (4);

\draw[edge] (3) -- (7);

\draw[edge] (1) -- (4);

\draw[edge] (7) -- (4);

\draw[edge] (5) -- (1);
\draw[edge] (5) -- (2);

\draw[edge] (1) -- (7);

        \end{tikzpicture}
       
   }

\vspace{0.4cm}

    \subfloat[$H_{5}$]{
     \begin{tikzpicture}[
    scale=0.35,
    transform shape,
    every node/.style={font=\Huge}
]

\tikzset{
    vertex/.style={
        draw,
        circle,
        fill=black,
        minimum size=0.4cm,
        inner sep=0pt
    },
    edge/.style={
        thick,
        postaction={decorate},
        decoration={
            markings,
            mark=at position 0.5 with {\arrow{>}}
        }
    }
}

\node[vertex, label=above right:$6$] (6) at (0, 0) {};
\node[vertex, label=above:$4$] (4) at (0, -2) {};
\node[vertex, label=right:$2$] (2) at (2.2, -3) {};
\node[vertex, label=left:$7$] (7) at (-3.5, 1) {};
\node[vertex, label=right:$5$] (5) at (3.5, 1) {};
\node[vertex, label=above:$1$] (1) at (0, 3.5) {};
\node[vertex, label=left:$3$] (3) at (-2.2, -3) {};

\draw[edge] (6) -- (1);
\draw[edge] (6) -- (2);
\draw[edge] (6) -- (3);
\draw[edge] (6) -- (7);

\draw[edge] (5) -- (7);

\draw[edge] (1) -- (3);
\draw[edge] (2) -- (3);

\draw[edge] (5) -- (2);

\draw[edge] (2) -- (4);
\draw[edge] (3) -- (4);

\draw[edge] (5) -- (1);

\draw[edge] (5) -- (6);

\draw[edge] (1) -- (7);

\end{tikzpicture}
        
    }
    \hspace{0.22cm}
    \subfloat[$H_{6}$]{
   \begin{tikzpicture}[
    scale=0.35,
    transform shape,
    every node/.style={font=\Huge}
]

\tikzset{
    vertex/.style={
        draw,
        circle,
        fill=black,
        minimum size=0.4cm,
        inner sep=0pt
    },
    edge/.style={
        thick,
        postaction={decorate},
        decoration={
            markings,
            mark=at position 0.5 with {\arrow{>}}
        }
    }
}

\node[vertex, label=above:$1$] (1) at (0, 3.5) {};
\node[vertex, label=above right:$6$] (6) at (0, 0) {};
\node[vertex, label=above:$4$] (4) at (0, -2) {};
\node[vertex, label=right:$2$] (2) at (2.2, -3) {};
\node[vertex, label=left:$7$] (7) at (-3.5, 1) {};
\node[vertex, label=right:$5$] (5) at (3.5, 1) {};
\node[vertex, label=left:$3$] (3) at (-2.2, -3) {};


\draw[edge] (6) -- (1);
\draw[edge] (6) -- (2);
\draw[edge] (6) -- (3);
\draw[edge] (6) -- (7);

\draw[edge] (1) -- (3);

\draw[edge] (2) -- (3);

\draw[edge] (3) -- (4);

\draw[edge] (2) -- (4);

\draw[edge] (5) -- (1);
\draw[edge] (5) -- (2);

\draw[edge] (5) -- (6);

\draw[edge] (1) -- (7);

\end{tikzpicture}
      
    }
\hspace{0.24cm}
\subfloat[$H_{7}$]{
\begin{tikzpicture}[
    scale=0.35,
    transform shape,
    every node/.style={font=\Huge}
]

\tikzset{
    vertex/.style={
        draw,
        circle,
        fill=black,
        minimum size=0.4cm,
        inner sep=0pt
    },
    edge/.style={
        thick,
        postaction={decorate},
        decoration={
            markings,
            mark=at position 0.5 with {\arrow{>}}
        }
    }
}

\node[vertex, label=above:$1$] (1) at (0, 2.5) {};
\node[vertex, label=above right:$2$] (2) at (0, -0.5) {};
\node[vertex, label=right:$3$] (3) at (3, 0) {};
\node[vertex, label=left:$4$] (4) at (-3, -2.5) {};
\node[vertex, label=right:$5$] (5) at (0, -2.5) {};
\node[vertex, label=right:$6$] (6) at (3, -2.5) {};
\node[vertex, label=below:$7$] (7) at (0, -4.3) {};


\draw[edge] (1) -- (2);

\draw[edge] (1) -- (3);

\draw[edge] (1) -- (6);

\draw[edge] (1) -- (4);

\draw[edge] (2) -- (4);

\draw[edge] (2) -- (6);

\draw[edge] (2) -- (5);

\draw[edge] (3) -- (6);

\draw[edge] (4) -- (5);

\draw[edge] (4) -- (7);

\draw[edge] (5) -- (7);

\draw[edge] (6) -- (7);

\end{tikzpicture}
       
    }
\vspace{0.3cm}

   \subfloat[$H_{8}$]{
      \begin{tikzpicture}[
    scale=0.35,
    transform shape,
    every node/.style={font=\Huge}
]

\tikzset{
    vertex/.style={
        draw,
        circle,
        fill=black,
        minimum size=0.4cm,
        inner sep=0pt
    },
    edge/.style={
        thick,
        postaction={decorate},
        decoration={
            markings,
            mark=at position 0.5 with {\arrow{>}}
        }
    }
}

\node[vertex, label=above:$1$] (1) at (0, 2.5) {};
\node[vertex, label=above right:$2$] (2) at (0, -0.5) {};
\node[vertex, label=right:$3$] (3) at (3, 0) {};
\node[vertex, label=left:$4$] (4) at (-3, -2.5) {};
\node[vertex, label=right:$5$] (5) at (0, -2.5) {};
\node[vertex, label=right:$6$] (6) at (3, -2.5) {};
\node[vertex, label=below:$7$] (7) at (0, -4.3) {};


\draw[edge] (1) -- (2);

\draw[edge] (1) -- (3);

\draw[edge] (1) -- (6);

\draw[edge] (1) -- (4);

\draw[edge] (2) -- (4);

\draw[edge] (2) -- (6);

\draw[edge] (2) -- (5);

\draw[edge] (3) -- (6);

\draw[edge] (4) -- (5);

\draw[edge] (4) -- (7);

\draw[edge] (6) -- (7);

\end{tikzpicture}
     
}
\hspace{0.3cm}
\subfloat[$H_{9}$]{
    \begin{tikzpicture}[
    scale=0.35,
    transform shape,
    every node/.style={font=\Huge}
]

\tikzset{
    vertex/.style={
        draw,
        circle,
        fill=black,
        minimum size=0.4cm,
        inner sep=0pt
    },
    edge/.style={
        thick,
        postaction={decorate},
        decoration={
            markings,
            mark=at position 0.5 with {\arrow{>}}
        }
    }
}

\node[vertex, label=above:$1$] (1) at (0, 2.5) {};
\node[vertex, label=above right:$2$] (2) at (0, -0.5) {};
\node[vertex, label=right:$3$] (3) at (3, 0) {};
\node[vertex, label=left:$4$] (4) at (-3, -2.5) {};
\node[vertex, label=right:$5$] (5) at (0, -2.5) {};
\node[vertex, label=right:$6$] (6) at (3, -2.5) {};
\node[vertex, label=below:$7$] (7) at (0, -4.3) {};


\draw[edge] (1) -- (2);

\draw[edge] (1) -- (3);

\draw[edge] (1) -- (6);

\draw[edge] (1) -- (4);

\draw[edge] (2) -- (4);

\draw[edge] (2) -- (6);

\draw[edge] (2) -- (5);

\draw[edge] (3) -- (6);

\draw[edge] (4) -- (5);

\draw[edge] (4) -- (7);

\end{tikzpicture}
      
   }
\hspace{0.3cm}
    \subfloat[$H_{10}$]{
      \begin{tikzpicture}[
    scale=0.35,
    transform shape,
    every node/.style={font=\Huge}
]

\tikzset{
    vertex/.style={
        draw,
        circle,
        fill=black,
        minimum size=0.4cm,
        inner sep=0pt
    },
    edge/.style={
        thick,
        postaction={decorate},
        decoration={
            markings,
            mark=at position 0.5 with {\arrow{>}}
        }
    }
}

\node[vertex, label=above:$1$] (1) at (0, 2.5) {};
\node[vertex, label=above right:$2$] (2) at (0, -0.5) {};
\node[vertex, label=right:$3$] (3) at (3, 0) {};
\node[vertex, label=left:$4$] (4) at (-3, -2.5) {};
\node[vertex, label=right:$5$] (5) at (0, -2.5) {};
\node[vertex, label=right:$6$] (6) at (3, -2.5) {};
\node[vertex, label=below:$7$] (7) at (0, -4.3) {};


\draw[edge] (1) -- (2);

\draw[edge] (1) -- (3);

\draw[edge] (1) -- (6);

\draw[edge] (1) -- (4);

\draw[edge] (2) -- (4);

\draw[edge] (2) -- (6);

\draw[edge] (2) -- (5);

\draw[edge] (3) -- (6);

\draw[edge] (4) -- (5);

\draw[edge] (4) -- (7);

\draw[edge] (5) -- (7);

\end{tikzpicture}
        
    } 
 \hspace{0.3cm}
  \subfloat[$H_{11}$]{
\begin{tikzpicture}[
    scale=0.35,
    transform shape,
    every node/.style={font=\Huge}
]

\tikzset{
    vertex/.style={
        draw,
        circle,
        fill=black,
        minimum size=0.4cm,
        inner sep=0pt
    },
    edge/.style={
        thick,
        postaction={decorate},
        decoration={
            markings,
            mark=at position 0.5 with {\arrow{>}}
        }
    }
}

\node[vertex, label=right:$1$] (1) at (1.15, -2.5) {};
\node[vertex, label=right:$2$] (2) at (2.3, 0.6) {};
\node[vertex, label=left:$4$] (4) at (-3.7, 0.6) {};
\node[vertex, label=above:$3$] (3) at (-0.7, 3) {};
\node[vertex, label=left:$5$] (5) at (-2.95, -2.5) {};
\node[vertex, label=left:$6$] (6) at (-5.2, -4.7) {};
\node[vertex, label=right:$7$] (7) at (3.4, -4.7) {};


\draw[edge] (1) -- (2);

\draw[edge] (3) -- (2);

\draw[edge] (3) -- (4);

\draw[edge] (1) -- (5);

\draw[edge] (3) -- (1);

\draw[edge] (4) -- (5);

\draw[edge] (2) -- (5);

\draw[edge] (3) -- (5);

\draw[edge] (1) -- (4);

\draw[edge] (1) -- (6);

\draw[edge] (4) -- (6);

\draw[edge] (2) -- (7);

\draw[edge] (5) -- (7);

\end{tikzpicture}

  }
   \hspace{0.3cm}
  \subfloat[$G_{2}^{9}$]{
\begin{tikzpicture}[
    scale=0.35,
    transform shape,
    every node/.style={font=\Huge}
]

\tikzset{
    vertex/.style={
        draw,
        circle,
        fill=black,
        minimum size=0.4cm,
        inner sep=0pt
    },
    edge/.style={
        thick,
        postaction={decorate},
        decoration={
            markings,
            mark=at position 0.5 with {\arrow{>}}
        }
    }
}

\node[vertex, label=left:$1$] (1) at (-0.5, 1) {};
\node[vertex, label=above:$2$] (2) at (2.2, -0.25) {};
\node[vertex, label=right:$4$] (4) at (6, 1) {};
\node[vertex, label=left:$3$] (3) at (-1.1, -1.4) {};
\node[vertex, label=above:$5$] (5) at (3, 3.5) {};
\node[vertex, label=right:$6$] (6) at (3, -2.7) {};


\draw[edge] (1) -- (3);

\draw[edge] (1) -- (4);

\draw[edge] (1) -- (5);

\draw[edge] (1) -- (6);

\draw[edge] (2) -- (4);

\draw[edge] (2) -- (6);

\draw[edge] (3) -- (6);

\draw[edge] (4) -- (6);

\draw[edge] (4) -- (5);

\end{tikzpicture}

  }
   \caption{Semi-transitive orientations of individual graphs in Figure~\ref{fig:min-non-comp-all} ($H_{1}$ is not semi-transitive, and $G_{2}^{9}$ is a special instance of the family $G_{n}^{9}$, which is semi-transitive).}
    \label{fig:min-non-comp-part4-semi-trans}
\end{figure}
Consider the orientations of the graphs in Figure~\ref{fig:min-non-comp-part4-semi-trans}.  In $H_{2}$, the longest directed path is $6 \rightarrow 1 \rightarrow 4 \rightarrow 7$, and the vertices $6$ and $7$ are not adjacent. All other directed paths in $H_{2}$ have length at most $2$.  In $H_{3}$, the only directed paths of length greater than $2$ whose endpoints are adjacent are $3 \rightarrow 7 \rightarrow 4 \rightarrow 2$ and $1 \rightarrow 3 \rightarrow 7 \rightarrow 4$. In each case, the vertices involved induce a clique.  In $H_{4}$, there are no directed paths of length greater than $2$ whose endpoints are adjacent.  In $H_{5}$, the only directed path of length greater than $2$ with adjacent endpoints is $5 \rightarrow 6 \rightarrow 1 \rightarrow 7$, and the corresponding vertices induce a clique.  For the graphs $H_{6}, H_{7}, H_{8}, H_{9},$ and $H_{10}$, there are no directed paths of length greater than $2$ whose endpoints are adjacent.  In $H_{11}$, the only directed paths of length greater than $2$ with adjacent endpoints are $3 \rightarrow 1 \rightarrow 2 \rightarrow 5$ and $3 \rightarrow 1 \rightarrow 4 \rightarrow 5$. In each case, the vertices involved induce a clique.  Finally, for the graph $G_{2}^{9}$, a special case of the family $G_{n}^{9}$ with $n=2$, there are no directed paths of length greater than $2$.  

Hence, in all cases, semi-transitivity can be verified by direct inspection.

\subsection{Minimal Non-Comparability Graphs that are Non-Word-Representable}
\label{section-minimal-non-semi-trans}

In this subsection, we establish the non-word-representability of all graphs listed in Theorem~\ref{non-wr-graphs-thm} via a sequence of intermediate propositions.

\begin{proposition}
    \label{theorem-g9n}
    Every graph in the family \( G_{n}^{9} \), for \( n \geq 3 \), as shown in Figure~\ref{fig:min-non-comp-all}, is minimal non-word-representable.
\end{proposition}

\begin{proof}
Let \( G \in G_{n}^{9} \) with \( n \geq 3 \). Suppose, for the sake of contradiction, that \( G \) is word-representable. Then, by Theorem~\ref{wrg=semi}, it admits a semi-transitive orientation. By Theorem~\ref{semi-v-source}, we fix such an orientation in which the vertex \( d \) is a source. By symmetry between \( a \) and \( b \), we may further assume that the edge \( \{a,b\} \) is oriented as \( a \rightarrow b \).

Now, for the edges \( \{a,c\} \) and \( \{c,b\} \), the orientations \( a \rightarrow c \) and \( c \rightarrow b \) are not possible, since the path \( d \rightarrow a \rightarrow c \rightarrow b \) would create a shortcut, violating semi-transitivity. Similarly, the orientations \( b \rightarrow c \) and \( c \rightarrow a \) are not possible, as the path \( d \rightarrow b \rightarrow c \rightarrow a \) creates a shortcut. Thus, the only possibilities are: \( a \rightarrow c, \ b \rightarrow c \) or \( c \rightarrow a, \ c \rightarrow b \). We analyze both cases.

In either case, the edge \( \{1,b\} \) must be oriented as \( 1 \rightarrow b \), since \( b \rightarrow 1 \) would create a path \( d \rightarrow a \rightarrow b \rightarrow 1 \), which forms a shortcut.

\paragraph{Case 1. \( a \rightarrow c, \ b \rightarrow c \)}
\begin{observation}
\label{itoaitobcase1}
In any semi-transitive orientation of \( G \) in which \( d \) is a source and the edges are oriented as \( a \rightarrow c \), \( b \rightarrow c \), and \( a \rightarrow b \), it follows that for every vertex \( i \), \( 2 \leq i \leq n-1 \), the edges incident to \( a \) and \( b \) must be oriented as
\( i \rightarrow a \) and \( i \rightarrow b \).
\end{observation}

\begin{proof}
Suppose first that \( a \rightarrow i \) and \( i \rightarrow b \). Then the directed path \( a \rightarrow i \rightarrow b \rightarrow c \) forms a shortcut, which is not allowed in a semi-transitive orientation. Similarly, if \( b \rightarrow i \) and \( i \rightarrow a \), then the path \( b \rightarrow i \rightarrow a \rightarrow c \) also forms a shortcut. Thus, the only remaining possibilities are: (i) \( a \rightarrow i \) and \( b \rightarrow i \), and (ii) \( i \rightarrow a \) and \( i \rightarrow b \).

Consider the case \( i = 2 \). If (i) holds, then in particular \( b \rightarrow 2 \), and hence the path \( d \rightarrow 1 \rightarrow b \rightarrow 2 \) forms a shortcut. Therefore, we must have \( 2 \rightarrow a \) and \( 2 \rightarrow b \). Now assume, for the sake of contradiction, that there exists some \( i \) such that \( a \rightarrow i \) and \( b \rightarrow i \). Let \( i \) be the smallest such index. Then the path \( d \rightarrow (i-1) \rightarrow b \rightarrow i \) forms a shortcut, again contradicting semi-transitivity. Hence, for all \( i \), \( 2 \leq i \leq n-1 \), we must have \( i \rightarrow a \) and \( i \rightarrow b \).
\end{proof}
Finally, consider the edge $\{a, n\}$. If \( a \rightarrow n \), then \( d \rightarrow (n-1) \rightarrow a \rightarrow n \) violates semi-transitivity. If \( n \rightarrow a \), then \( d \rightarrow n \rightarrow a \rightarrow b \) violates semi-transitivity. In both cases we reach a contradiction.  Thus, no semi-transitive orientation exists in Case~1.

\paragraph{Case 2. \( c \rightarrow a, \ c \rightarrow b \)}

\begin{observation}
\label{itoaitobcase2}
In any semi-transitive orientation of \( G \) in which \( d \) is a source and the edges are oriented as \( c \rightarrow a \), \( c \rightarrow b \), and \( a \rightarrow b \), it follows that for every vertex \( i \), \( 2 \leq i \leq n-1 \), the edges incident to \( a \) and \( b \) must be oriented as $i \rightarrow a \quad \text{and} \quad i \rightarrow b.$
\end{observation}

\begin{proof}
Suppose first that \( a \rightarrow i \) and \( i \rightarrow b \). Then the directed path \( c \rightarrow a \rightarrow i \rightarrow b \) forms a shortcut, which is not allowed in a semi-transitive orientation. Similarly, if \( b \rightarrow i \) and \( i \rightarrow a \), then the path \( c \rightarrow b \rightarrow i \rightarrow a \) also forms a shortcut. Thus, the only remaining possibilities are: (i) \( a \rightarrow i \) and \( b \rightarrow i \), and (ii) \( i \rightarrow a \) and \( i \rightarrow b \).

Consider the case \( i = 2 \). If (i) holds, then in particular \( b \rightarrow 2 \), and hence the path \( d \rightarrow 1 \rightarrow b \rightarrow 2 \) forms a shortcut. Therefore, we must have \( 2 \rightarrow a \) and \( 2 \rightarrow b \). Now assume, for the sake of contradiction, that there exists some \( i \) such that \( a \rightarrow i \) and \( b \rightarrow i \). Let \( i \) be the smallest such index. Then the path \( d \rightarrow (i-1) \rightarrow b \rightarrow i \) forms a shortcut, again contradicting semi-transitivity. Hence, for all \( i \), \( 2 \leq i \leq n-1 \), we must have \( i \rightarrow a \) and \( i \rightarrow b \).
\end{proof}

Now when we consider the edge $\{a,n\}$, both orientations \( a \rightarrow n \) and \( n \rightarrow a \) creates shortcuts, via the paths \( d \rightarrow (n-1) \rightarrow a \rightarrow n \) and \( d \rightarrow n \rightarrow a \rightarrow b \), respectively. Thus, no semi-transitive orientation exists in Case~2 also. Since no semi-transitive orientation exists in either case, we arrive at a contradiction. Hence, \( G \) is non-word-representable. Moreover, as \( G \) is minimal non-comparability, every proper induced subgraph of \( G \) is comparability, and hence word-representable. Therefore, \( G \) is minimal non-word-representable.
\end{proof}

\begin{proposition}
    \label{theorem-g4n}
    Every graph in the family \( G_{n}^{4} \), as shown in Figure~\ref{fig:min-non-comp-all}, is minimal non-word-representable.
\end{proposition}

\begin{proof}

Consider an arbitrary graph \( G \in G_{n}^{4} \), where \( n \geq 3 \). Suppose, for the sake of contradiction, that \( G \) is word-representable. By Theorem~\ref{semi-v-source}, if a graph is word-representable, it admits a semi-transitive orientation in which any arbitrary vertex can be selected as a source. Therefore, we fix a semi-transitive orientation \( O \) of the original graph \( G \) in which vertex \( 1 \) is chosen to be a source. Let \( G' \) be the induced subgraph of \( G \) obtained by removing the vertices \( x \) and \( y \), and let \( O' \) be the restriction of the orientation \( O \) to \( G' \). Since removing vertices cannot introduce new incoming directed edges to the remaining vertices, the restricted orientation \( O' \) is guaranteed to be a valid semi-transitive orientation of \( G' \) in which vertex \( 1 \) remains a source. We next analyze the possible configurations of \( O' \).

The edge $\{2,3\}$ can be oriented either $2 \rightarrow 3$ or $3 \rightarrow 2$; we consider these as two separate cases. Irrespective of the orientation of $\{2,3\}$, the following holds regarding the neighborhood of vertex $2n+1$.
\begin{observation}
\label{obs:2nplus1_orientation}
In \( O' \), the orientation of all edges incident to vertex $2n+1$ is determined as follows:
\begin{enumerate}
   
    \item If $(2n+1) \rightarrow 2n$, then $(2n+1) \rightarrow i$ for all $i \in \{2, \dots, 2n-1\}$.
    \item If $2n \rightarrow (2n+1)$, then $i \rightarrow (2n+1)$ for all $i \in \{2, \dots, 2n-1\}$.
\end{enumerate}
\end{observation}

\begin{proof}
Suppose first that the edge $\{2n, 2n+1\}$ is oriented as $(2n+1) \to 2n$. We determine the orientation of each edge $\{2n+1, i\}$, where $i \in \{2, \dots, 2n-1\}$. Consider $i = 2$. If $2 \to 2n+1$, then the directed path $1 \to 2 \to 2n+1 \to 2n$ forms a shortcut, contradicting semi-transitivity. Hence $(2n+1) \to 2$. Now let $j$ satisfy $2 \le j \le 2n-2$ and assume $(2n+1) \to j$. If $j+1 \to 2n+1$, then the directed path $1 \to j+1 \to 2n+1 \to 2$ forms a shortcut. Therefore $(2n+1) \to j+1$. By induction on $j$, it follows that $(2n+1) \to i$ for all $2 \le i \le 2n-1$.

Now suppose instead that the edge $\{2n, 2n+1\}$ is oriented as $2n \to 2n+1$. Consider $i = 2$. If $(2n+1) \to 2$, then the directed path $1 \to 2n \to 2n+1 \to 2$ forms a shortcut. Hence $2 \to 2n+1$. Now let $j$ satisfy $2 \le j \le 2n-2$ and assume $j \to 2n+1$. If $2n+1 \to j+1$, then the directed path $1 \to 2 \to 2n+1 \to j+1$ creates a shortcut. Therefore $j+1 \to 2n+1$. By induction, $i \to 2n+1$ for all $2 \le i \le 2n-1$.
\end{proof}

Now, let us look at the case where we orient the edge $\{2,3\}$ as $2 \rightarrow 3$.
\paragraph{Case 1. $2 \rightarrow 3$}
\begin{observation}
\label{case1lemma_i_i+1}
In \( O' \), if the edge \( \{2,3\} \) is oriented as \( 2 \rightarrow 3 \), then for every even vertex \( i \), \( 4 \leq i \leq 2n - 2 \), the orientations \( i \rightarrow (i-1) \) and \( i \rightarrow (i+1) \) are forced by semi-transitivity.
\end{observation}

\begin{proof}
We proceed by induction on even indices \( i \), where \( 4 \le i \le 2n-2 \).

\textbf{Base case (\( i = 4 \)):}
If \( 3 \rightarrow 4 \), then the path \( 1 \rightarrow 2 \rightarrow 3 \rightarrow 4 \) forms a shortcut, contradicting semi-transitivity. Hence, \( 4 \rightarrow 3 \) must hold. Similarly, if \( 5 \rightarrow 4 \), then, since \( 1 \) is a source, we have \( 1 \rightarrow 5 \), and hence the path \( 1 \rightarrow 5 \rightarrow 4 \rightarrow 3 \) forms a shortcut, contradicting semi-transitivity. Hence, \( 4 \rightarrow 5 \). Therefore, the claim holds for \( i = 4 \).

\textbf{Inductive step:}
Assume that for some even integer \( j \), where \( 4 \le j \le 2n-4 \), we have \( j \rightarrow (j-1) \) and \( j \rightarrow (j+1) \). We show that the statement holds for \( i = j+2 \), which is also even and satisfies \( 6 \le i \le 2n-2 \). If \( (j+1) \rightarrow (j+2) \), then the path \( 1 \rightarrow j \rightarrow (j+1) \rightarrow (j+2) \) forms a shortcut, contradicting semi-transitivity. Hence, \( (j+2) \rightarrow (j+1) \). Similarly, if \( (j+3) \rightarrow (j+2) \), then the path \( 1 \rightarrow (j+1) \rightarrow (j+3) \rightarrow (j+2) \) forms a shortcut, contradicting semi-transitivity. Hence, \( (j+2) \rightarrow (j+3) \). Thus, the claim holds for \( i = j+2 \). Hence, by induction, the observation holds for all even \( i \) with \( 4 \le i \le 2n-2 \).
\end{proof}

Note that if \( (2n-1) \rightarrow 2n \), then the path \( 1 \rightarrow (2n-2) \rightarrow (2n-1) \rightarrow 2n \) forms a shortcut, contradicting semi-transitivity. Thus, the valid orientation is \( 2n \rightarrow (2n-1) \).

At this point, all possible semi-transitive orientations under Case~1 have been determined; there are exactly two such orientations, illustrated in Figures~\ref{G''UA} and \ref{G''UB}. This corresponds precisely to fixing the orientation of the edge \( \{2,3\} \) as \( 2 \rightarrow 3 \) and applying the two cases described in Observation~\ref{obs:2nplus1_orientation}.

Now, let us look at the case where we orient the edge $\{2,3\}$ as $3 \rightarrow 2$.

\paragraph{Case 2. $3 \rightarrow 2$}
\begin{observation}
\label{case2lemma_i_i+1}
In \( O' \), if the edge \( \{2,3\} \) is oriented as \( 3 \rightarrow 2 \), then for every even vertex \( i \), \( 4 \leq i \leq 2n - 2 \), the orientations \( (i-1) \rightarrow i \) and \( (i+1) \rightarrow i \) are forced by semi-transitivity.
\end{observation}

\begin{proof}
We proceed by induction on even indices \( i \), where \( 4 \le i \le 2n-2 \).

\textbf{Base case (\( i = 4 \)):}
If \( 4 \rightarrow 3 \), then the directed path
\( 1 \rightarrow 4 \rightarrow 3 \rightarrow 2 \)
violates semi-transitivity. Hence, \( 3 \rightarrow 4 \) must hold. Similarly, if \( 4 \rightarrow 5 \), then the directed path
\( 1 \rightarrow 3 \rightarrow 4 \rightarrow 5 \)
violates semi-transitivity, so the correct orientation is \( 5 \rightarrow 4 \). Therefore, the claim holds for \( i = 4 \).

\textbf{Inductive step:}
Assume that for some even integer \( j \), where \( 4 \le j \le 2n-4 \), the orientations \( (j-1) \rightarrow j \) and \( (j+1) \rightarrow j \) hold. We show that the statement holds for \( i = j+2 \), which is also even and satisfies \( 6 \le i \le 2n-2 \). If \( (j+2) \rightarrow (j+1) \), then the directed path
\( 1 \rightarrow (j+2) \rightarrow (j+1) \rightarrow j \)
violates semi-transitivity. Hence, \( (j+1) \rightarrow (j+2) \). Similarly, if \( (j+2) \rightarrow (j+3) \), then the directed path
\( 1 \rightarrow (j+1) \rightarrow (j+2) \rightarrow (j+3) \)
violates semi-transitivity. Hence, \( (j+3) \rightarrow (j+2) \). Thus, the claim holds for \( i = j+2 \). By induction, the observation holds for all even \( i \) with \( 4 \le i \le 2n-2 \).
\end{proof}

Note that  if \( 2n \rightarrow (2n-1) \), the path \( 1 \rightarrow 2n \rightarrow (2n-1) \rightarrow (2n-2) \) violates semi-transitivity. Thus, the valid orientation is \( (2n-1) \rightarrow 2n \).

At this point, all possible semi-transitive orientations under Case $2$ have been determined; there are exactly two such orientations, illustrated in Figures~\ref{G''UC} and \ref{G''UD}. This corresponds precisely to fixing the orientation of the edge \( \{2,3\} \) as \( 3 \rightarrow 2 \) and applying the two cases described in Observation~\ref{obs:2nplus1_orientation}.

\begin{figure}[htbp]
    \centering
    
    \tikzset{
        vertex/.style = {circle, draw=black, thick, fill=black, text=white, inner sep=0pt, minimum size=0.45cm},
        edge/.style = {->, line width=1.2pt, shorten >=2pt, shorten <=2pt},
        redge/.style = {<-, line width=1.2pt, shorten >=2pt, shorten <=2pt}
    }

\subfloat[Case $(a)$: $2 \rightarrow3$, $(2n+1) \rightarrow 2n$.]{
  \begin{tikzpicture}[scale=0.5, transform shape, every node/.style={font=\Huge}, >=Stealth]

    \node[vertex, label=right:{$1$}] (1) at (1.15, -2.5) {};
    \node[vertex, label=right:{$2$}] (2) at (3.4, 0) {};
    \node[vertex, label=left:{$2n$}] (2n) at (-3.9, 0) {};
    \node[vertex, label=right:{$3$}] (3) at (3.37, 2.3) {};
    \node[vertex, label=left:{$2n-1$}] (r1) at (-3.9, 2.3) {};
    \node[vertex, label=right:{$4$}] (4) at (1.2, 4) {};
    \node[vertex] (r) at (-1.7, 4) {};
    \node[vertex, label=left:{$2n+1$}] (2n+1) at (-1.65, -2.5) {};

    \draw[edge] (1) -- (2);
    \draw[edge] (2) -- (3);
    \draw[redge] (3) -- (4);
    \draw[edge] (4) -- (r) node[midway, above, yshift=5pt] {. . . . . . .};

    \draw[edge] (1) -- (3);
    \draw[edge] (1) -- (4);
    \draw[edge] (1) -- (r);
    \draw[edge] (1) -- (r1);
    \draw[edge] (1) -- (2n);
    \draw[edge] (1) -- (2n+1);

    \draw[edge] (2n+1) -- (2);
    \draw[edge] (2n+1) -- (3);
    \draw[edge] (2n+1) -- (4);
    \draw[edge] (2n+1) -- (r);
    \draw[edge] (2n+1) -- (r1);
    \draw[edge] (2n+1) -- (2n);

    \draw[edge] (r) -- (r1);
    \draw[edge] (2n) -- (r1);

  \end{tikzpicture}
 \label{G''UA}
}
\hspace{2.2cm}
\subfloat[Case $(b)$: $2 \rightarrow3$, $2n \rightarrow (2n+1)$.]{
  \begin{tikzpicture}[scale=0.5, transform shape, every node/.style={font=\Huge}, >=Stealth]

    \node[vertex, label=right:{$1$}] (1) at (1.15, -2.5) {};
    \node[vertex, label=right:{$2$}] (2) at (3.4, 0) {};
    \node[vertex, label=left:{$2n$}] (2n) at (-3.9, 0) {};
    \node[vertex, label=right:{$3$}] (3) at (3.37, 2.3) {};
    \node[vertex, label=left:{$2n-1$}] (r1) at (-3.9, 2.3) {};
    \node[vertex, label=right:{$4$}] (4) at (1.2, 4) {};
    \node[vertex] (r) at (-1.7, 4) {};
    \node[vertex, label=left:{$2n+1$}] (2n+1) at (-1.65, -2.5) {};

    \draw[edge] (1) -- (2);
    \draw[edge] (2) -- (3);
    \draw[redge] (3) -- (4);
    \draw[edge] (4) -- (r) node[midway, above, yshift=5pt] {. . . . . . .};

    \draw[edge] (1) -- (3);
    \draw[edge] (1) -- (4);
    \draw[edge] (1) -- (r);
    \draw[edge] (1) -- (r1);
    \draw[edge] (1) -- (2n);
    \draw[edge] (1) -- (2n+1);

    \draw[redge] (2n+1) -- (2);
    \draw[redge] (2n+1) -- (3);
    \draw[redge] (2n+1) -- (4);
    \draw[redge] (2n+1) -- (r);
    \draw[redge] (2n+1) -- (r1);
    \draw[redge] (2n+1) -- (2n);

    \draw[edge] (r) -- (r1);
    \draw[edge] (2n) -- (r1);

  \end{tikzpicture}
 \label{G''UB}
}

\vspace{0.2cm}

\subfloat[Case $(c)$: $3 \rightarrow 2$, $(2n+1) \rightarrow 2n$.]{
  \begin{tikzpicture}[scale=0.5, transform shape, every node/.style={font=\Huge}, >=Stealth]

    \node[vertex, label=right:{$1$}] (1) at (1.15, -2.5) {};
    \node[vertex, label=right:{$2$}] (2) at (3.4, 0) {};
    \node[vertex, label=left:{$2n$}] (2n) at (-3.9, 0) {};
    \node[vertex, label=right:{$3$}] (3) at (3.37, 2.3) {};
    \node[vertex, label=left:{$2n-1$}] (r1) at (-3.9, 2.3) {};
    \node[vertex, label=right:{$4$}] (4) at (1.2, 4) {};
    \node[vertex] (r) at (-1.7, 4) {};
    \node[vertex, label=left:{$2n+1$}] (2n+1) at (-1.65, -2.5) {};

    \draw[edge] (1) -- (2);
    \draw[edge] (3) -- (2);
    \draw[edge] (3) -- (4);
    \draw[redge] (4) -- (r) node[midway, above, yshift=5pt] {. . . . . . .};

    \draw[edge] (1) -- (3);
    \draw[edge] (1) -- (4);
    \draw[edge] (1) -- (r);
    \draw[edge] (1) -- (r1);
    \draw[edge] (1) -- (2n);
    \draw[edge] (1) -- (2n+1);

    \draw[edge] (2n+1) -- (2);
    \draw[edge] (2n+1) -- (3);
    \draw[edge] (2n+1) -- (4);
    \draw[edge] (2n+1) -- (r);
    \draw[edge] (2n+1) -- (r1);
    \draw[edge] (2n+1) -- (2n);

    \draw[redge] (r) -- (r1);
    \draw[redge] (2n) -- (r1);

  \end{tikzpicture}
 \label{G''UC}
}
\hspace{2cm}
\subfloat[Case $(d)$: $3 \rightarrow2$, $2n \rightarrow (2n+1)$.]{
  \begin{tikzpicture}[scale=0.5, transform shape, every node/.style={font=\Huge}, >=Stealth]

    \node[vertex, label=right:{$1$}] (1) at (1.15, -2.5) {};
    \node[vertex, label=right:{$2$}] (2) at (3.4, 0) {};
    \node[vertex, label=left:{$2n$}] (2n) at (-3.9, 0) {};
    \node[vertex, label=right:{$3$}] (3) at (3.37, 2.3) {};
    \node[vertex, label=left:{$2n-1$}] (r1) at (-3.9, 2.3) {};
    \node[vertex, label=right:{$4$}] (4) at (1.2, 4) {};
    \node[vertex] (r) at (-1.7, 4) {};
    \node[vertex, label=left:{$2n+1$}] (2n+1) at (-1.65, -2.5) {};

    \draw[edge] (1) -- (2);
    \draw[edge] (3) -- (2);
    \draw[edge] (3) -- (4);
    \draw[redge] (4) -- (r) node[midway, above, yshift=5pt] {. . . . . . .};

    \draw[edge] (1) -- (3);
    \draw[edge] (1) -- (4);
    \draw[edge] (1) -- (r);
    \draw[edge] (1) -- (r1);
    \draw[edge] (1) -- (2n);
    \draw[edge] (1) -- (2n+1);

    \draw[redge] (2n+1) -- (2);
    \draw[redge] (2n+1) -- (3);
    \draw[redge] (2n+1) -- (4);
    \draw[redge] (2n+1) -- (r);
    \draw[redge] (2n+1) -- (r1);
    \draw[redge] (2n+1) -- (2n);

    \draw[redge] (r) -- (r1);
    \draw[redge] (2n) -- (r1);

  \end{tikzpicture}
 \label{G''UD}
}

\caption{All possible semi-transitive orientations of \( O' \)}
\label{fig:4orientationsofG''U}
\end{figure}

Hence, there are exactly four possible semi-transitive orientations of \( O' \) as illustrated in Figure \ref{fig:4orientationsofG''U}. We now show that none of these orientations can be extended to the semi-transitive orientation \( O \) of the original graph; that is, after adding the vertices $x$ and $y$ together with all edges incident to $G'$, each resulting orientation fails to satisfy the semi-transitivity conditions.

We verify that none of the four orientations  shown in Figure \ref{fig:4orientationsofG''U} can extend to \( O \) once vertices \( x \) and \( y \) are restored:

\begin{enumerate}
    \item Extending the graph in Figure \ref{G''UA}. 
    \par
    The edge $\{2n, x\}$ cannot be oriented as $2n \rightarrow x$, since the directed path $1 \rightarrow (2n+1) \rightarrow 2n \rightarrow x$ forms a shortcut. Likewise, it cannot be oriented as $x \rightarrow 2n$, since the path $1 \rightarrow x \rightarrow 2n \rightarrow (2n-1)$ forms a shortcut. Therefore, no semi-transitive extension to $G$ exists in this case. The orientation of $\{1, x\}$ is fixed as $1 \rightarrow x$.

    \vspace{0.2cm}
    \item Extending the graph in Figure \ref{G''UB}. 
    \par
    The edges \( \{2, y\} \) and \( \{2n+1, y\} \) admit four possible orientations. In each case, a shortcut is formed:
    \begin{enumerate}
        \item $2 \rightarrow y$, $y \rightarrow (2n+1)$: $1 \rightarrow 2 \rightarrow y \rightarrow (2n+1)$.
        \item $2 \rightarrow y$, $(2n+1) \rightarrow y$: $2 \rightarrow 3 \rightarrow (2n+1) \rightarrow y$.
        \item $y \rightarrow 2$, $y \rightarrow (2n+1)$: $y \rightarrow 2 \rightarrow 3 \rightarrow (2n+1)$.
        \item $y \rightarrow 2$, $(2n+1) \rightarrow y$: $1 \rightarrow (2n+1) \rightarrow y \rightarrow 2$.
    \end{enumerate}
    Hence, no semi-transitive extension exists.

    \vspace{0.2cm}
    \item Extending the graph in Figure \ref{G''UC}. 
    \par
    As above, all four orientations of \( \{2, y\} \) and \( \{2n+1, y\} \) result in a shortcut:
    \begin{enumerate}
        \item $2 \rightarrow y$, $y \rightarrow (2n+1)$: $1 \rightarrow 2 \rightarrow y \rightarrow (2n+1)$.
        \item $2 \rightarrow y$, $(2n+1) \rightarrow y$: $(2n+1) \rightarrow 3 \rightarrow 2 \rightarrow y$.
        \item $y \rightarrow 2$, $y \rightarrow (2n+1)$: $y \rightarrow (2n+1) \rightarrow 3 \rightarrow 2$.
        \item $y \rightarrow 2$, $(2n+1) \rightarrow y$: $1 \rightarrow (2n+1) \rightarrow y \rightarrow 2$.
    \end{enumerate}
    Thus, no semi-transitive extension exists.

    \vspace{0.2cm}
    \item Extending the graph in Figure \ref{G''UD}. 
    \par
    The edge $\{2n, x\}$ cannot be oriented as $2n \rightarrow x$, since $1 \rightarrow (2n-1) \rightarrow 2n \rightarrow x$ forms a shortcut. Similarly, $x \rightarrow 2n$ yields the shortcut $1 \rightarrow x \rightarrow 2n \rightarrow (2n+1)$. Hence, no semi-transitive extension exists in this case.
\end{enumerate}
In each of the four cases, no semi-transitive extension to $G$ exists (That is, we do not get the semi-transitive orientation \( O \)). Hence, $G$ is not word-representable. Since $G$ is a minimal non-comparability graph, every proper induced subgraph of $G$ is a comparability graph and hence word-representable. Therefore, $G$ is minimal non-word-representable.
\end{proof}

\begin{remark}
The graph $H_{1}$ is known to be non-word-representable; see \cite{GraphsCapturingAlternation}, 
where it appears under the name Co-$T_{2}$.
\end{remark}

With this, the proofs of Theorem \ref{wr-graphs-thm} and Theorem \ref{non-wr-graphs-thm} are concluded.

\medskip
\noindent
As a byproduct of the classification obtained in Theorems~\ref{wr-graphs-thm} and \ref{non-wr-graphs-thm}, we derive in the next subsection a characterization of minimal non-word-representable graphs with an all-adjacent vertex. In particular, these graphs arise by adjoining an all-adjacent vertex to the minimal non-comparability graphs that are word-representable.

\subsection{Minimal Non-Word-Representable Graphs with an All-Adjacent Vertex}

The minimal non-comparability graphs that are word-representable are precisely those listed in Theorem~\ref{wr-graphs-thm}. We show that adjoining an all-adjacent vertex to each of these graphs yields exactly the class of minimal non-word-representable graphs containing an all-adjacent vertex.

By Theorem~\ref{lemma_comp_wrg}, a graph $G$ is word-representable if and only if $H = G - x$, where $x$ is an all-adjacent vertex, is a comparability graph. Combining this with the classification of minimal non-comparability graphs with respect to word-representability, we obtain a complete characterization of minimal non-word-representable graphs containing an all-adjacent vertex. The following theorem formalizes this characterization.

\begin{theorem}
\label{minimal-nwrg-all-adj-thm}
Let $G$ be a graph on $n$ vertices, and let $x \in V(G)$ be a vertex of degree $n-1$. Let $H = G - x$. Then $G$ is minimal non-word-representable if and only if:
\begin{enumerate}
    \item $H$ is a minimal non-comparability graph, and
    \item $H$ is word-representable.
\end{enumerate}
\end{theorem}

\begin{proof}
Suppose that $H$ is a minimal non-comparability graph and is word-representable. Since $H$ is not a comparability graph, Theorem~\ref{lemma_comp_wrg} implies that $G$ is not word-representable. We now establish minimality. Let $H'$ be any proper induced subgraph of $G$. If $x \in V(H')$, then $H' - x$ is a proper induced subgraph of $H$. By minimality of $H$, the graph $H' - x$ is a comparability graph, and hence Theorem~\ref{lemma_comp_wrg} implies that $H'$ is word-representable. If $x \notin V(H')$, then $H' \subset H$. Since $H$ is word-representable, it follows that $H'$ is word-representable. Thus, every proper induced subgraph of $G$ is word-representable. Hence $G$ is minimal non-word-representable.

Conversely, suppose that $G$ is minimal non-word-representable and contains an all-adjacent vertex $x$. Let $H = G - x$. By minimality of $G$, every proper induced subgraph of $G$ is word-representable, and in particular $H$ is word-representable. If $H$ were a comparability graph, then Theorem~\ref{lemma_comp_wrg} would imply that $G$ is word-representable, a contradiction. Thus, $H$ is not a comparability graph. Finally, suppose that $H$ is not minimal non-comparability. Then there exists a proper induced subgraph $H_2 \subset H$ that is also non-comparability. The induced subgraph $H_2 \cup \{x\}$ is then not word-representable by Theorem~\ref{lemma_comp_wrg}, contradicting the minimality of $G$. Therefore, $H$ is a minimal non-comparability graph.
\end{proof}
We have identified all minimal non-comparability graphs that are word-representable. By Theorem~\ref{minimal-nwrg-all-adj-thm}, adjoining an all-adjacent vertex to each of these graphs yields precisely the class of minimal non-word-representable graphs containing an all-adjacent vertex.

In the next section, we turn to the second problem stated in the Introduction.

\section{Covering Word-Representable Graphs by Comparability Graphs}
\label{covering-section}

In this section, we study the relationship between word-representable graphs and comparability graphs from a covering perspective, which is quantitative in nature, complementing our earlier study on identifying common minimal induced subgraphs of these two classes. Specifically, we study the minimum number of comparability graphs required to cover the edge set of a word-representable graph. This corresponds to Problem~2 in the introduction, which was posed in~\cite{kenkireth2026word}.

We now introduce the notion of the cover number of a graph by a class of graphs.

\begin{definition}
\label{general-cover-number-defn}
Let $\mathcal{F}$ be a class of graphs. The \emph{cover number of a graph $G$ by $\mathcal{F}$}, denoted by $\text{cov}_{\mathcal{F}}(G)$, is the minimum integer $k$ for which there exist graphs $G_1, \dots, G_k \in \mathcal{F}$ such that
\[
E(G) = \bigcup_{i=1}^k E(G_i).
\]
\end{definition}

As a special case, if $\mathcal{F}$ is the class of comparability graphs, we denote the cover number of the graph $G$ by comparability graphs as $\text{cov}_{\text{comp}}(G)$. Similarly, for the class of perfect graphs, we denote it by $\text{cov}_{\text{perf}}(G)$. The cover number $\text{cov}_{\mathcal{F}}(G)$ provides a quantitative measure of the extent to which a graph $G$ deviates from the structural properties of the class $\mathcal{F}$. In the context of this section, where $G$ is a word-representable graph and $\mathcal{F}$ is the class of comparability graphs, $\text{cov}_{\text{comp}}(G)$ quantifies the distance of $G$ from the class of comparability graphs; larger values indicate greater structural deviation.

\begin{figure}[htbp]
\centering

\tikzset{
    vertex/.style = {shape=circle},
    undirected/.style = {thick},
    directed/.style = {
        thick,
        postaction={decorate},
        decoration={
            markings,
            mark=at position 0.5 with {\arrow{>}}
        }
    }
}

\subfloat[]{
\centering
\begin{tikzpicture}[scale=0.48, transform shape, every node/.style={font=\Large}]

\node[vertex,fill=black,text=white,minimum size=0.35cm,label=left:{$1$}] (a1) at (-1.5,-1.5) {};
\node[vertex,fill=black,text=white,minimum size=0.35cm,label=above:{$2$}] (a2) at (0.15,0) {};
\node[vertex,fill=black,text=white,minimum size=0.35cm,label=right:{$3$}] (a3) at (1.5,-1.5) {};
\node[vertex,fill=black,text=white,minimum size=0.35cm,label=left:{$4$}] (a4) at (-1.5,-3.5) {};
\node[vertex,fill=black,text=white,minimum size=0.35cm,label=right:{$5$}] (a5) at (1.5,-3.5) {};

\draw[undirected] (a1)--(a2)--(a3)--(a5)--(a4)--(a1);

\end{tikzpicture}
}
\hfill
\subfloat[]{
\centering
\begin{tikzpicture}[scale=0.48, transform shape, every node/.style={font=\Large}]

\node[vertex,fill=black,text=white,minimum size=0.35cm,label=left:{$1$}] (b1) at (-1.5,-1.5) {};
\node[vertex,fill=black,text=white,minimum size=0.35cm,label=above:{$2$}] (b2) at (0.15,0) {};
\node[vertex,fill=black,text=white,minimum size=0.35cm,label=right:{$3$}] (b3) at (1.5,-1.5) {};
\node[vertex,fill=black,text=white,minimum size=0.35cm,label=left:{$4$}] (b4) at (-1.5,-3.5) {};
\node[vertex,fill=black,text=white,minimum size=0.35cm,label=right:{$5$}] (b5) at (1.5,-3.5) {};

\draw[directed] (b2)--(b3);
\draw[directed] (b5)--(b3);
\draw[directed] (b5)--(b4);

\end{tikzpicture}
}
\hfill
\subfloat[]{
\centering
\begin{tikzpicture}[scale=0.48, transform shape, every node/.style={font=\Large}]

\node[vertex,fill=black,text=white,minimum size=0.35cm,label=left:{$1$}] (c1) at (-1.5,-1.5) {};
\node[vertex,fill=black,text=white,minimum size=0.35cm,label=above:{$2$}] (c2) at (0.15,0) {};
\node[vertex,fill=black,text=white,minimum size=0.35cm,label=right:{$3$}] (c3) at (1.5,-1.5) {};
\node[vertex,fill=black,text=white,minimum size=0.35cm,label=left:{$4$}] (c4) at (-1.5,-3.5) {};
\node[vertex,fill=black,text=white,minimum size=0.35cm,label=right:{$5$}] (c5) at (1.5,-3.5) {};

\draw[directed] (c1)--(c2);
\draw[directed] (c1)--(c4);

\end{tikzpicture}
}

\caption{Representation of $C_5$ as a union of two comparability graphs}
\label{fig:c5_decomposition}
\end{figure}
\begin{example}
If $G$ is a comparability graph, then its cover number by comparability graphs is $1$, since $G$ itself suffices. On the other hand, if $G$ is a word-representable graph that is not a comparability graph, then its cover number is at least $2$. The cycle $C_5$ is not a comparability graph. However, it can be expressed as the union of two comparability graphs, as illustrated in Figure~\ref{fig:c5_decomposition}. Hence, the cover number of $C_5$ by comparability graphs is $2$. 
\end{example}

Marits \cite{MARITS2026400} derived an exact formula for the cover number of a graph $G$ by the class of graphs satisfying $\chi(H) = \omega(H)$. This result, given as Corollary 4 in \cite{MARITS2026400}, is stated below.

\begin{theorem}[\cite{MARITS2026400}]
\label{cover-formula}
For every graph $G$, the cover number of $G$ by the class of graphs $\mathcal{C} = \{H \mid \chi(H) = \omega(H)\}$ is given by
\[
\left\lceil \frac{\log \chi(G)}{\log \omega(G)} \right\rceil.
\]
\end{theorem}

Given that comparability graphs form a subclass of perfect graphs, we can formalize a lower bound for $\text{cov}_{\text{comp}}(G)$ using the bounds established by Marits \cite{MARITS2026400}. Theorem 1 in \cite{MARITS2026400} states that if $\mathcal{P}$ and $\mathcal{Q}$ are graph classes such that $\mathcal{P} \subseteq \mathcal{Q}$, then their respective cover numbers, denoted by $\text{cov}_{\mathcal{P}}(G)$ and $\text{cov}_{\mathcal{Q}}(G)$, satisfy $\text{cov}_{\mathcal{P}}(G) \ge \text{cov}_{\mathcal{Q}}(G)$ for any graph $G$. Applying this to the subclass relationship between comparability and perfect graphs \cite{golumbic2004algorithmic}, it immediately follows that $\text{cov}_{\text{comp}}(G) \ge \text{cov}_{\text{perf}}(G)$.

Furthermore, Theorem 6 in \cite{MARITS2026400} establishes that $\text{cov}_{\text{perf}}(G)$ is bounded below by the value in Theorem~\ref{cover-formula}. Combining these facts yields the following lower bound for the cover number of a graph $G$ by comparability graphs:
\[
\text{cov}_{\text{comp}}(G) \ge \text{cov}_{\text{perf}}(G) \ge \left\lceil \frac{\log \chi(G)}{\log \omega(G)} \right\rceil.
\]

We now turn to circle graphs, a subclass of word-representable graphs, to establish an upper bound on their cover number by comparability graphs.

\begin{definition}
A graph $G$ is a \emph{circle graph} if its vertices can be represented as chords of a circle, where two vertices are adjacent if and only if their corresponding chords intersect.
\end{definition}

Circle graphs are known to be $\chi$-bounded; that is, there exists a function $f$ such that $\chi(G) \le f(\omega(G))$ for every circle graph $G$. The following result provides an explicit upper bound for this function:

\begin{theorem}[\cite{davies2022improved}]
\label{thm:circle-chi-bound}
Every circle graph $G$ with clique number $\omega(G)=\omega$ satisfies
\[
\chi(G) \le 2\omega \log \omega + 2\omega \log(\log \omega) + 10\omega.
\]
\end{theorem}

Since every bipartite graph is a comparability graph \cite{golumbic2004algorithmic}, the cover number of any graph $G$ by comparability graphs is bounded above by its cover number by bipartite graphs:
\[
\operatorname{cov}_{\text{comp}}(G) \le \operatorname{cov}_{\text{bip}}(G).
\]
It is a well-established result by Harary et al. \cite{harary1977biparticity} that the number of bipartite graphs required to cover the edges of a graph $G$ is precisely $\lceil \log \chi(G) \rceil$. Therefore, we have the following inequality:
\[
\operatorname{cov}_{\text{comp}}(G) \le \lceil \log \chi(G) \rceil.
\]
By substituting the chromatic number bound from Theorem~\ref{thm:circle-chi-bound} into this expression, we obtain:
\[
\operatorname{cov}_{\text{comp}}(G) \le \left\lceil \log \left( 2\omega \log \omega + 2\omega \log(\log \omega) + 10\omega \right) \right\rceil.
\]
To determine the asymptotic behavior as $\omega \to \infty$, we observe that the expression inside the logarithm is dominated by the term $\omega \log \omega$. Consequently, the upper bound simplifies as follows:
\[
\operatorname{cov}_{\text{comp}}(G) = O(\log(\omega \log \omega)) = O(\log \omega + \log \log \omega) = O(\log \omega).
\]

Hence, for a circle graph $G$, the cover number $\operatorname{cov}_{\text{comp}}(G)$ is $O(\log \omega(G))$.

We now consider the case where $\omega(G) \le 2$, which corresponds precisely to triangle-free circle graphs. 

\begin{theorem}
Let $G$ be a triangle-free circle graph. Then the cover number of $G$ by comparability graphs is at most $3$, and this bound is tight.
\end{theorem}

\begin{proof}
It is a well-established result by Kostochka \cite{kostochka1988upper} that every triangle-free circle graph $G$ satisfies $\chi(G) \le 5$. The number of bipartite graphs required to cover the edges of a graph $G$ is given by $\lceil \log_2 \chi(G) \rceil$ \cite{harary1977biparticity}. Since every bipartite graph is a comparability graph, the cover number of $G$ by comparability graphs satisfies
\[
\text{cov}_{\text{comp}}(G) \le \text{cov}_{\text{bip}}(G) = \lceil \log_2 \chi(G) \rceil.
\] where $\text{cov}_{\text{bip}}(G)$ denotes the cover number of $G$ by bipartite graphs.
Substituting $\chi(G) \le 5$ yields $\text{cov}_{\text{comp}}(G) \le \lceil \log_2 5 \rceil = 3$. 

To show the bound is tight, consider the triangle-free circle graph constructed by Ageev \cite{ageev1996triangle}, which has $\chi(G) = 5$. Since $\omega(G) = 2$, we apply the lower bound from Theorem \ref{cover-formula}:
\[
\text{cov}_{\text{comp}}(G) \ge \left\lceil \frac{\log_2 \chi(G)}{\log_2 \omega(G)} \right\rceil = \left\lceil \frac{\log_2 5}{\log_2 2} \right\rceil = 3.
\]
This confirms that the bound is tight.
\end{proof}

We now prove a logarithmic lower bound on the cover number by comparability graphs for word-representable graphs. The proof relies on Theorem~\ref{thm:suk-tomon}, due to Suk and Tomon, regarding the chromatic number of Hasse diagrams.

It is shown in \cite{book} (Theorem 4.2.14) that every Hasse diagram (cover graph of a poset) is a word-representable graph, specifically establishing that the class of triangle-free word-representable graphs is exactly the class of cover graphs of posets.

The following result establishes the existence of Hasse diagrams with large chromatic number.

\begin{theorem}[\cite{Suk-Tomon}]
\label{thm:suk-tomon}
For every positive integer $n$, there exists a Hasse diagram $G$ on $n$ vertices with chromatic number $\chi(G) = \Omega(n^{1/4})$.
\end{theorem}

\begin{theorem}
\label{triangle-free-infinite-proof-logn}
There exist word-representable graphs on $n$ vertices whose cover number by comparability graphs is $\Omega(\log n)$.
\end{theorem}

\begin{proof}
For any sufficiently large $n$, let $G_n$ be a Hasse diagram on $n$ vertices as guaranteed by Theorem~\ref{thm:suk-tomon}, satisfying $\chi(G_n) = \Omega(n^{1/4})$. As Hasse diagrams are triangle-free, we have $\omega(G_n) \le 2$. Since $\chi(G_n) \to \infty$ as $n \to \infty$, $G_n$ must contain at least one edge for sufficiently large $n$, ensuring $\omega(G_n) = 2$.

Applying the lower bound established by Marits \cite{MARITS2026400}, the cover number satisfies:
\[
\operatorname{cov}_{\text{comp}}(G_n) \ge \left\lceil \frac{\log \chi(G_n)}{\log \omega(G_n)} \right\rceil = \lceil \log \chi(G_n) \rceil.
\]
Substituting $\chi(G_n) = \Omega(n^{1/4})$, there exists a constant $c > 0$ such that $\chi(G_n) \ge c n^{1/4}$ for all sufficiently large $n$. Consequently:
\[
\log \chi(G_n) \ge \log(c n^{1/4}) = \frac{1}{4} \log n + \log c = \Omega(\log n).
\]
Thus, $\operatorname{cov}_{\text{comp}}(G_n) = \Omega(\log n)$. Since Hasse diagrams are word-representable \cite{book}, the result follows.
\end{proof}

\begin{remark}
The edge set of any graph $G$ on $n$ vertices can be covered by $O(\log n)$ comparability graphs. This follows because the edge set of $G$ can be covered by $\lceil \log_2 \chi(G) \rceil$ bipartite graphs \cite{harary1977biparticity}, and every bipartite graph is a comparability graph. Since $\chi(G) \le n$, the cover number is $O(\log n)$. Combined with Theorem~\ref{triangle-free-infinite-proof-logn}, this shows that the $O(\log n)$ upper bound on the cover number by comparability graphs is asymptotically tight for the class of word-representable graphs.
\end{remark}

We now identify several subclasses of word-representable graphs such that each graph in these classes has cover number by comparability graphs at most $2$.  We use standard definitions of outer-planar graphs, $2$-trees, sub-cubic graphs, and triangle-free planar graphs (see \cite{book}). It is known that these classes are all subclasses of word-representable graphs; see \cite{book} for details.

\begin{theorem}
For each graph $G$ in the following classes of word-representable graphs, the cover number of $G$ by comparability graphs is at most $2$. \begin{enumerate}
    \item outer-planar graphs,
    \item $2$-trees,
    \item sub-cubic graphs,
    \item triangle-free planar graphs.
\end{enumerate}
\end{theorem}

\begin{proof}
We first note that every bipartite graph is a comparability graph \cite{golumbic2004algorithmic}. Hence, $\text{cov}_{\text{comp}}(G) \le \text{cov}_{\text{bip}}(G) $. We now use the fact that the number of bipartite graphs required to cover the edge set of a graph $G$ is exactly $\lceil \log \chi(G) \rceil$, as given by \cite{harary1977biparticity}. Hence, $\text{cov}_{\text{comp}}(G) \le \lceil \log \chi(G) \rceil$ . It is known that outer-planar graphs, $2$-trees, and triangle-free planar graphs are $3$-colorable (see \cite{book}), while sub-cubic graphs satisfy $\chi(G) \le 4$. In all cases,
$
\lceil \log \chi(G) \rceil \le 2.
$
Therefore, every graph $G$ in these classes satisfies $\text{cov}_{\text{comp}}(G) \le 2$.
\end{proof}

\section{Concluding remarks}
\label{conclusion-section}
In this work, we determine the intersection of minimal non-comparability graphs and minimal non-word-representable graphs. This intersection consists of two infinite families of graphs, namely $G_n^9$ for $n \geq 3$ and $G_n^4$, together with the graph $H_1$, all illustrated in Figure~\ref{fig:min-non-comp-all}. This indicates that the classes of minimal non-comparability graphs and minimal non-word-representable graphs are largely disjoint. This result is obtained by classifying minimal non-comparability graphs into those that are word-representable and those that are not. As a consequence of this classification, we completely characterize the set of minimal non-word-representable graphs containing an all-adjacent vertex.

Furthermore, we investigate the cover number by comparability graphs within the class of word-representable graphs. We demonstrate the existence of word-representable graphs on $n$ vertices whose cover number by comparability graphs is $\Omega(\log n)$. Since any $n$-vertex graph can be covered by $O(\log n)$ comparability graphs, this upper bound is asymptotically tight for the class of word-representable graphs. This indicates that there exist word-representable graphs exhibiting the maximum possible structural deviation from comparability graphs under this covering measure. Additionally, we establish an upper bound of $O(\log \omega(G))$ for the cover number by comparability graphs of any circle graph $G$. For triangle-free circle graphs, we establish that the cover number by comparability graphs is at most $3$ and demonstrate that this bound is tight. Finally, we identify four subclasses of word-representable graphs for which every graph in these classes has a cover number by comparability graphs of at most $2$.

\bibliography{sn-bibliography}

@article{kitaev2023humanverifiable,
author = {Kitaev, Sergey and Sun, Haoran},
year = {2024},
month = {03},
pages = {9},
title = {Human-verifiable proofs in the theory of word-representable graphs},
volume = {58},
journal = {RAIRO - Theoretical Informatics and Applications},
doi = {10.1051/ita/2024004}
}

@article{kostochka1988upper,
  title={Upper bounds on the chromatic number of graphs},
  author={Kostochka, Alexandr V},
  journal={Trudy Inst. Mat.(Novosibirsk)},
  volume={10},
  number={Modeli i Metody Optim.},
  pages={204--226},
  year={1988}
}

@article{Suk-Tomon,
author = {Suk, Andrew and Tomon, István},
title = {Hasse diagrams with large chromatic number},
journal = {Bulletin of the London Mathematical Society},
volume = {53},
number = {3},
pages = {747-758},
doi = {https://doi.org/10.1112/blms.12457},
url = {https://londmathsoc.onlinelibrary.wiley.com/doi/abs/10.1112/blms.12457},
eprint = {https://londmathsoc.onlinelibrary.wiley.com/doi/pdf/10.1112/blms.12457},
year = {2021}
}

@book{golumbic2004algorithmic,
author = {Golumbic, Martin Charles},
title = {Algorithmic Graph Theory and Perfect Graphs (Annals of Discrete Mathematics, Vol 57)},
year = {2004},
isbn = {0444515305},
publisher = {North-Holland Publishing Co.},
address = {NLD}
}

@InProceedings{GraphsCapturingAlternation,
author="Halld{\'o}rsson, Magn{\'u}s M.
and Kitaev, Sergey
and Pyatkin, Artem",
editor="Gao, Yuan
and Lu, Hanlin
and Seki, Shinnosuke
and Yu, Sheng",
title="Graphs Capturing Alternations in Words",
booktitle="Developments in Language Theory",
year="2010",
publisher="Springer Berlin Heidelberg",
address="Berlin, Heidelberg",
pages="436--437",
isbn="978-3-642-14455-4"
}

@article{onrepgraphs,
author = {Kitaev, Sergey and Pyatkin, Artem},
year = {2008},
month = {01},
pages = {45-54},
title = {On Representable Graphs.},
volume = {13},
journal = {Journal of Automata, Languages and Combinatorics}
}

@article{davies2022improved,
  title={Improved bounds for colouring circle graphs},
  author={Davies, James},
  journal={Proceedings of the American Mathematical Society},
  volume={150},
  number={12},
  pages={5121--5135},
  year={2022}
}

@article{ageev1996triangle,
  title={A triangle-free circle graph with chromatic number 5},
  author={Ageev, Alexander A},
  journal={Discrete Mathematics},
  volume={152},
  number={1-3},
  pages={295--298},
  year={1996},
  publisher={Elsevier}
}

@inproceedings{kenkireth2026word,
  title={On Word-Representability of Minimal Non-comparability Graphs},
  author={Kenkireth, Benny George and Sajith, Gopalan and Sasidharan, Sreyas},
  booktitle={Conference on Algorithms and Discrete Applied Mathematics},
  pages={250--263},
  year={2026},
  organization={Springer}
}

@article{harary1977biparticity,
  title={The biparticity of a graph},
  author={Harary, Frank and Hsu, Derbiau and Miller, Zevi},
  journal={Journal of graph theory},
  volume={1},
  number={2},
  pages={131--133},
  year={1977},
  publisher={Wiley Online Library}
}

@article{MARITS2026400,
title = {Cover numbers by certain graph families},
journal = {Discrete Applied Mathematics},
volume = {379},
pages = {400-404},
year = {2026},
issn = {0166-218X},
doi = {https://doi.org/10.1016/j.dam.2025.09.009},
url = {https://www.sciencedirect.com/science/article/pii/S0166218X25005359},
author = {Márton Marits}
}

@InProceedings{kitaev2017comprehensive,
author="Kitaev, Sergey",
editor="Charlier, {\'E}milie
and Leroy, Julien
and Rigo, Michel",
title="A Comprehensive Introduction to the Theory of Word-Representable Graphs",
booktitle="Developments in Language Theory",
year="2017",
publisher="Springer International Publishing",
address="Cham",
pages="36--67",
isbn="978-3-319-62809-7"
}

@article{newres,
author = {Collins, Andrew and Kitaev, Sergey and Lozin, Vadim},
year = {2014},
month = {11},
pages = {},
title = {New results on word-representable graphs},
volume = {216},
journal = {Discrete Applied Mathematics},
doi = {10.1016/j.dam.2014.10.024}
}

@book{book,
  title     = {Words and Graphs},
  author    = {Sergey Kitaev and Vadim Lozin},
  year      = {2015},
  publisher = {Springer},
  address   = {Cham},
  isbn      = {978-3-319-20829-4},
  doi       = {https://doi.org/10.1007/978-3-319-25859-1}
}

@article{Gallai1967TransitivOG,
  title={Transitiv orientierbare Graphen},
  author={T. Gallai},
  journal={Acta Mathematica Academiae Scientiarum Hungarica},
  year={1967},
  volume={18},
  pages={25-66},
  url={https://api.semanticscholar.org/CorpusID:119485995}
}

@InProceedings{Halldorsson2011,
author="Halld{\'o}rsson, Magn{\'u}s M.
and Kitaev, Sergey
and Pyatkin, Artem",
editor="Kolman, Petr
and Kratochv{\'i}l, Jan",
title="Alternation Graphs",
booktitle="Graph-Theoretic Concepts in Computer Science",
year="2011",
publisher="Springer Berlin Heidelberg",
address="Berlin, Heidelberg",
pages="191--202",
isbn="978-3-642-25870-1"
}

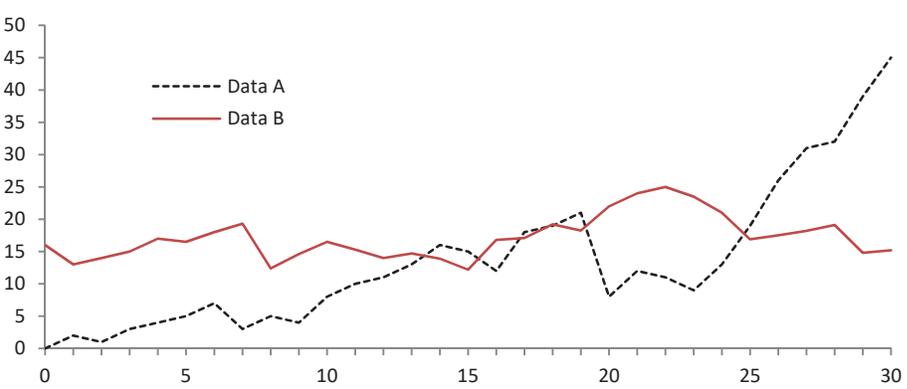
\end{document}